\documentclass[10pt,a4paper]{article}

\RequirePackage[colorlinks,hyperindex]{hyperref}

\usepackage{graphicx,subfigure}
\usepackage{subfigure}

\usepackage[square,numbers]{natbib}

\usepackage{pstricks,pst-plot}

\usepackage{amsmath}
\usepackage{amsthm}
\usepackage{amsfonts}
\usepackage{amssymb}
\usepackage{latexsym}
\usepackage{stackrel}
\usepackage[T1]{fontenc}
\usepackage{fourier}

\usepackage{lastpage}

\let\mathnumsetfont\mathbb
\newcommand\Rset{\mathnumsetfont R} 

\newtheorem{teo*}{Theorem}
\newtheorem{prop}{Proposition}

\def\barF{\overline{F}}

\begin{document}

\title{Extensions of the Burr Type XII distribution and Applications}
\author{Meitner Cadena \\
DECE, Universidad de las Fuerzas Armadas - ESPE \\
Sangolqu\'{\i}, Ecuador
}
\maketitle

\begin{abstract}
The Burr type XII (BXII) distribution has been largely used in different fields due to its great flexibility for fitting data.
These applications have typically involved data showing heavy-tailed behaviors.
In order to give more flexibility to the BXII distribution, in this paper, modifications to this distribution through the use of parametric functions are introduced.
For instance, members of this new family of distributions allow the analysis not only of data containing extreme values as the BXII distribution, but also of light-tailed data.
We refer to this new family of distributions as the extended Burr Type XII distribution (EBXIID) family.
Statistical properties of members of the EBXIID family are discussed.
The maximum likelihood method is proposed for estimating model parameters.
The performance of the new family of distributions is studied using simulations.
Applications of the new models to real data sets coming from different domains show that models of the EBXIID family are an alternative to other known distributions.

\begin{center}
\end{center}
\end{abstract}

\section{Introduction}

The called Burr Type XII (BXII) distribution belonging to a set of distributions introduced by Burr in 1942 \cite{Burr1942} has been largely studied due to the skill for fitting almost any given set of uni-modal data, see for instance \cite{Burr1942,Hatke1949,Burr1973,Rodriguez1977,LewisThesis}.
Further, an interesting feature of this distribution is its relation with other well-known distributions as the Lomax, compound Weibull, 
Weibull-exponential, logistic, log-logistic, and Weibull distributions and the Kappa family of distributions, see e.g. \cite{ZimmerKeatsWang1998,Tadikamalla1980}, which shows its wide flexibility for fitting different kinds of data.
With respect to applications, this distribution has been used in a variety of fields thanks to its coverage of a wide range of skewness and kurtosis values, for instance survival analysis, hydrology, environment and reliability \cite{WangCheng2010,SilvaOrtegaPaula2011}.

The survival function (sf) of the BXII distribution $F_B$ is the following two-parameter function:
\begin{equation}\label{burr}
\barF_B(x):=1-F_B(x) = \big(1+x^b\big)^{-a},\quad x>0,
\end{equation}
for some $a>0$ and $b>0$.
A common extension of its original formulation incorporates location and scale parameters, $\mu>0$ and $\sigma>0$ respectively, giving the sf
\begin{equation}\label{burrExt}
\barF_B(x) = \left(1+\left(\frac{x-\mu}{\sigma}\right)^b\right)^{-a},\quad x>\mu.
\end{equation}
The sfs defined as in (\ref{burr}) as well in (\ref{burrExt}) are regularly varying, i.e. they satisfy \cite{Karamata1930,Seneta,deHaan}, for some $\alpha\in\Rset$ called the tail index and for all $t>0$,
\begin{equation}\label{rv}
\lim_{x\to\infty}\frac{\barF_B(tx)}{\barF_B(x)}=t^{-\alpha}.
\end{equation}
Straightforward computations give that the tail index of both $\barF_B$ defined in (\ref{burr}) and in (\ref{burrExt}) is $ab$.
Note that is commonly mentioned that the shape parameters of $F_B$ are $a$ and $b$, see e.g. \cite{LioTsaiWu2010,DogruArslan2015,GanoraLaio2015}, but it is also $ab$, its tail index.
This result on the shape parameter of the BXII distribution shows that this distribution has a heavy-tailed tail, so this distribution may be used for modeling extreme values \cite{BinghamGoldieTeugels,resnick,deHaanFerreira}.

Furthermore, the BXII distribution has been used in combination with other distributions to generate new distributions.
Examples of this type of distributions are: 
the generalized log-Burr family, e.g. \cite{Lawless}, the Kumaraswamy Burr XII distribution introduced by Paranaiba et al. 
\cite{ParanaibaOrtegaCordeirodePascoa2013}, the Weibull Burr XII distribution introduced by Paranaiba et al. 
\cite{AfifyCordeiroOrtegaYousofButt2016}, and the McBXII distribution proposed by Gomes et al. \cite{GomesdaSilvaCordeiro2015}.

In this paper, (\ref{burr}) and (\ref{burrExt}) are extended in order to give them more flexibility for fitting data.
These modifications are done using combinations of appropriated functions and imposing suitable conditions to guarantee that the resulting functions be distributions.
Note that this configuration allows the introduction of parameters disposed in different ways, that will be adapted to the fit to given data.
This ability for fitting different kinds of data is analyzed using real data, namely by evaluating their performance with respect to competitors.
We refer to this wide distribution family as the extended Burr Type XII distribution (EBXIID) family.

The aim of this paper is two-fold.
First, to study statistical properties of members of the EBXIID family and methods for estimating the parameters of these members.
Second, to provide empirical evidence on the great flexibility of members of this family to fit real data from different domains.
For numerical assessments, implementations of these models were done using functions in the R software \cite{RCoreTeam}.

In the next section the EBXIID family is introduced and the members of this family to be studied are presented.
For these distributions conditions to guarantee that themselves are distributions are discussed.
Section \ref{sec2} presents statistical properties of these distributions.
Section \ref{sec3} is devoted to the maximum likelihood method for estimating the parameters of these distributions.
In Section \ref{sec4}, the performance of the parameter estimation method is studied using simulations.
Section \ref{sec5} shows applications of the new distributions to real data sets coming from different domains.
Section \ref{sec6} concludes the paper presenting discussions and conclusions and next further steps.
Proofs are presented in Annexe \ref{Proofs}.

\section{The EBXIID Family}

The cumulative distribution function (cdf) of the members of the EBXIID family is defined as follows.
\begin{equation}\label{defF}
F(x):=1-\big(1+g(x)\big)^{-a}, \quad x>d,
\end{equation}
where $a>0$, $d\in\Rset$ that depends on $g$, and $g$ is a positive, differentiable function defined on $(d,\infty)$, and satisfying $g'(x)\geq0$ and $g(x)\to\infty$ as $x\to\infty$.

The following result guarantees that $\barF:=1-F$ with $F$ defined in (\ref{defF}) is a sf.
\begin{prop}\label{prop:20170220}
$\barF$ with $F$ defined in (\ref{defF}) is a sf.
\end{prop}

$F$ defined in (\ref{defF}) contains a wide range of distributions, in particular those of (\ref{burr}) and (\ref{burrExt}).
This paper focuses on distributions where the associated functions $g$ are those specified in Table \ref{tab01}.
Among these functions the one corresponding to the BXII distribution is included, $g_0$.
Note that the considered variants of $g$ may give more flexibility to $g_0$, in particular by allowing modeling of as heavy- as well light-tailed data as will be shown later.

\begin{table}[!h]
\centering
\begin{tabular}{lccc}
\hline
\multicolumn{1}{c}{$g$} & Parameters & Conditions & Support \\
\hline
$g_0(x):=\big(c\,(x+d)\big)^b$ & $b>0$, $c>0$, $d\in\Rset$ & & $x>-d$ \\
$g_1(x):=\left(c\,\frac{x^\epsilon(x+d)^p}{\log(x+d)}\right)^b$ & $b>0$, $c>0$, $d\in\Rset$, $\epsilon\geq0$, & $x+d>1$, & $x>0$ \\
 & $p\geq0$ & $\log(x+d)\big[px+\epsilon(x+d)\big]\geq x$ & \\
$g_2(x):=\left(c\,\frac{x^\epsilon}{\log(x+d)}\right)^b\,e^{x^p}$ & $b\geq0$, $c>0$, $d\in\Rset$, $\epsilon\geq0$, & $x+d>1$, & $x>0$ \\
 & $p\geq0$ & $(x+d)\log(x+d)\left[b\epsilon+px^p\right]\geq bx$ & \\
$g_3(x):=\big(c\,\frac{x^\epsilon\log(x+d+1)}{\log(x+d)}\big)^b$ & $b>0$, $c>0$, $d\in\Rset$, $\epsilon\geq0$ & $x+d>1$, & $x>0$ \\
 & & $(x+d)\log(x+d)$ & \\
 & & $\times\left[\epsilon(x+d+1)\log(x+d+1)+1\right]$ & \\
 & & $\geq x(x+d+1)\log(x+d+1)$ & \\
\hline
\end{tabular}
\caption{Variants of functions $g$}
\label{tab01}
\end{table}
Parameters of $F$ defined in (\ref{defF}) with functions $g$ given in Table \ref{tab01}, may be fixed \emph{a priori} for defining sfs with less parameters. This happens when, for instance, it is fixed $\mu=0$ and $\sigma=1$ in $\barF_B$ defined in (\ref{burrExt}) to obtain $\barF_B$ defined in (\ref{burr}).
On the other hand, $g_1$, $g_2$ and $g_3$ are not derivations of $g_0$, implying that the distributions based on these functions $g$ are new distribution alternatives with respect to the BXII distribution and its extension (\ref{burr}) and (\ref{burrExt}), respectively.

The analysis of the shape parameter of $F$ when $g$ is one of the functions $g_1$, $g_2$ or $g_3$ may be done through the tail index by evaluating (\ref{rv}), but the computation of the involved limit is complex.
A way that facilitates such an analysis, is to use the following relationship given in \cite{Karamata1930} and e.g. \cite{cadenaThesis,NN2015CKanalysis,NN2015CKbis,CadenaKratzOmey2016}.
\begin{equation}\label{classM}
\lim_{x\to\infty}\frac{\log\barF(x)}{\log x}=-\alpha,
\end{equation}
where $\alpha$ corresponds to the tail index of $\barF$.
Table \ref{tab02} shows the tail indexes corresponding to the functions $g$ defined in Table \ref{tab01} computed by using (\ref{classM}).
The computation of the tail index associated to the distribution based on $g_0$ is included in that table.
This result was mentioned above.
\begin{table}[!h]
\centering
\begin{tabular}{ccc}
\hline
\multicolumn{1}{c}{$g$} & Tail index ($\alpha$) & Condition \\
\hline
$g_0$ & $ab$ \\
$g_1$ & $ab(\epsilon+p)$ \\
$g_2$ & $ab\epsilon$ & $p=0$ \\
 & $\infty$ & $p>0$ \\
$g_3$ & $ab\epsilon$ \\
\hline
\end{tabular}
\caption{Shape parameters of $F$ with $g$ defined in Table \ref{tab01}}
\label{tab02}
\end{table}

As example of these computations, the ones when $g=g_2$ are illustrated:
\begin{eqnarray*} 
\lim_{x\to\infty}\frac{\log\barF(x)}{\log x} & = & \lim_{x\to\infty}\frac{\log\left[\left(1+\left(c\,\frac{x^\epsilon}{\log(x+d)}\right)^be^{x^p}\right)^{-a}\right]}{\log x} \\
 & = & -a\times\lim_{x\to\infty}\frac{\log\left(1+\left(c\,\frac{x^\epsilon}{\log(x+d)}\right)^be^{x^p}\right)}{\log x} \\
 & = & -a\times\left[b\epsilon-b\times\lim_{x\to\infty}\frac{\log\big(\log(x+d)\big)}{\log x}+
\lim_{x\to\infty}\frac{x^p}{\log x}\right] \\
 & = & \left\{
\begin{array}{ll}
-ab\epsilon, & \textrm{if $p=0$} \\
-\infty, & \textrm{if $p>0$},
\end{array}
\right.
\end{eqnarray*} 
where the L'H\^{o}pital rule has been applied.

The results presented in Table \ref{tab02} show the large range that is covered for the tail index through the variants of $g$ presented in Table \ref{tab01}.
When this index is finite, $\barF$ is heavy-tailed, as mentioned earlier.
This is the case of $g_0$ (seen earlier), $g_1$, 
$g_2$ if $p=0$, and $g_3$.
On the other hand, for 
$g_2$ if $p>0$, the tail index corresponds to functions that are not so heavy-tailed.
This last feature that is not present in the BXII distribution gives to the EBXIID family more flexibility for fitting data.

Plots of cdfs of the members of the EBXII family that are studied in this paper are exhibited in Figure \ref{fig01}.
Each plot presents distributions associated to each one of the functions $g_0$, $g_1$, $g_2$ and $g_3$. 
These distributions are obtained by varying two of their parameters, $b$ and $c$, the others remaining fixed.
Through all these plots $x$ belonging to $(0,2)$ is considered.
All these plots show diverse behaviors of cdfs.
On the one hand, some curves rapidly approach $y=1$, meaning that the distributions involved tend to be light-tailed.
This is the case of 
$g_2$ with $p>0$.
These results are in line with the value of the tail indexes presented in Table \ref{tab02}.
On the other hand, the cdfs related with $g_0$, $g_1$ 
and $g_3$ are heavy-tailed, the more notorious being the cdfs involving $g_3$ because these curves seem to slowly converge to $y=1$.
Deepening in this behavior, it is found that the concerned tail index is $ab\epsilon\in\big\{2.4,3.6\big\}$, which is the lowest value that a tail index can reach among the cases studied.
For the examples presented, these values of tail indexes equal to the ones of the distribution that have associated $g_0$.
In the other case, $g_1$, its corresponding tail index, $\alpha_1$, is positive and, it is appreciated that the cdf approaches faster to $y=1$ since this index is higher, i.e. $\alpha_1>\alpha_0$ and $\alpha_1>\alpha_3$ with $\alpha_i$ the tail index of the cdf with $g=g_i$, $i$ = 0, 1, 2, 3.

\begin{figure}
\centering 
\subfigure[$g_0$]{
\includegraphics[scale=0.32]{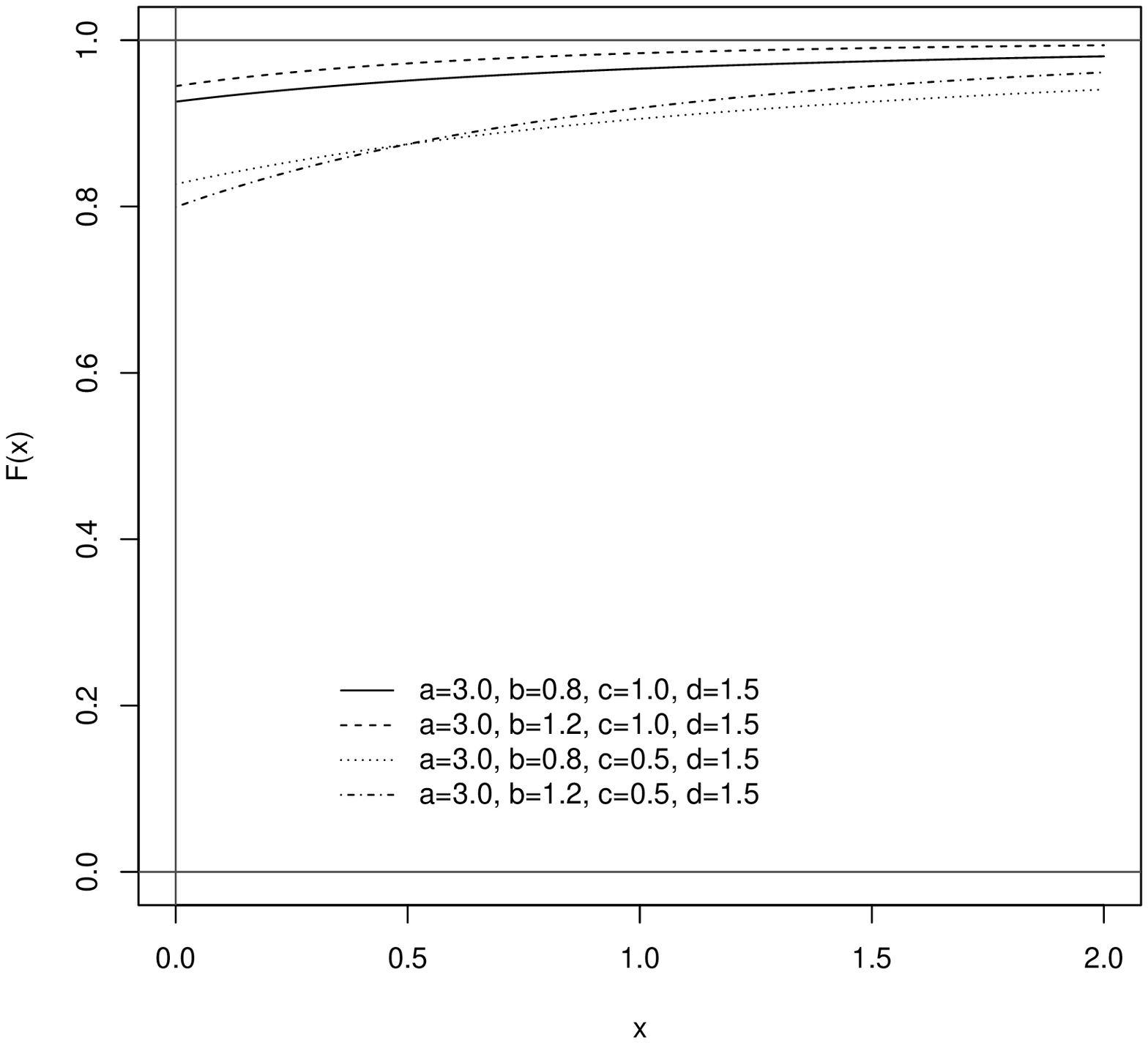} 
\label{fig01:subfig0}
}
\subfigure[$g_1$]{
\includegraphics[scale=0.32]{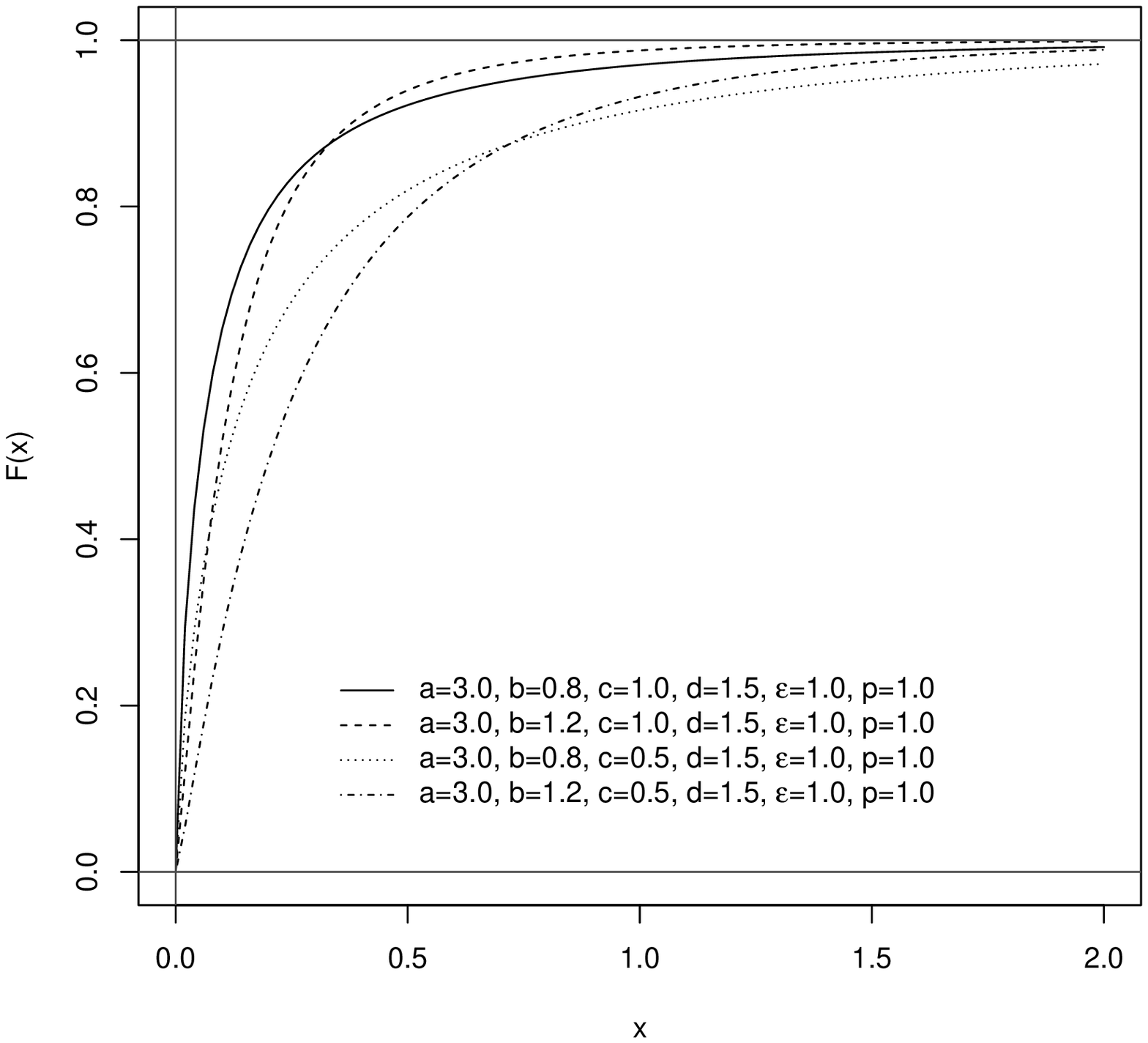} 
\label{fig01:subfig1}
}
\subfigure[$g_2$]{
\includegraphics[scale=0.32]{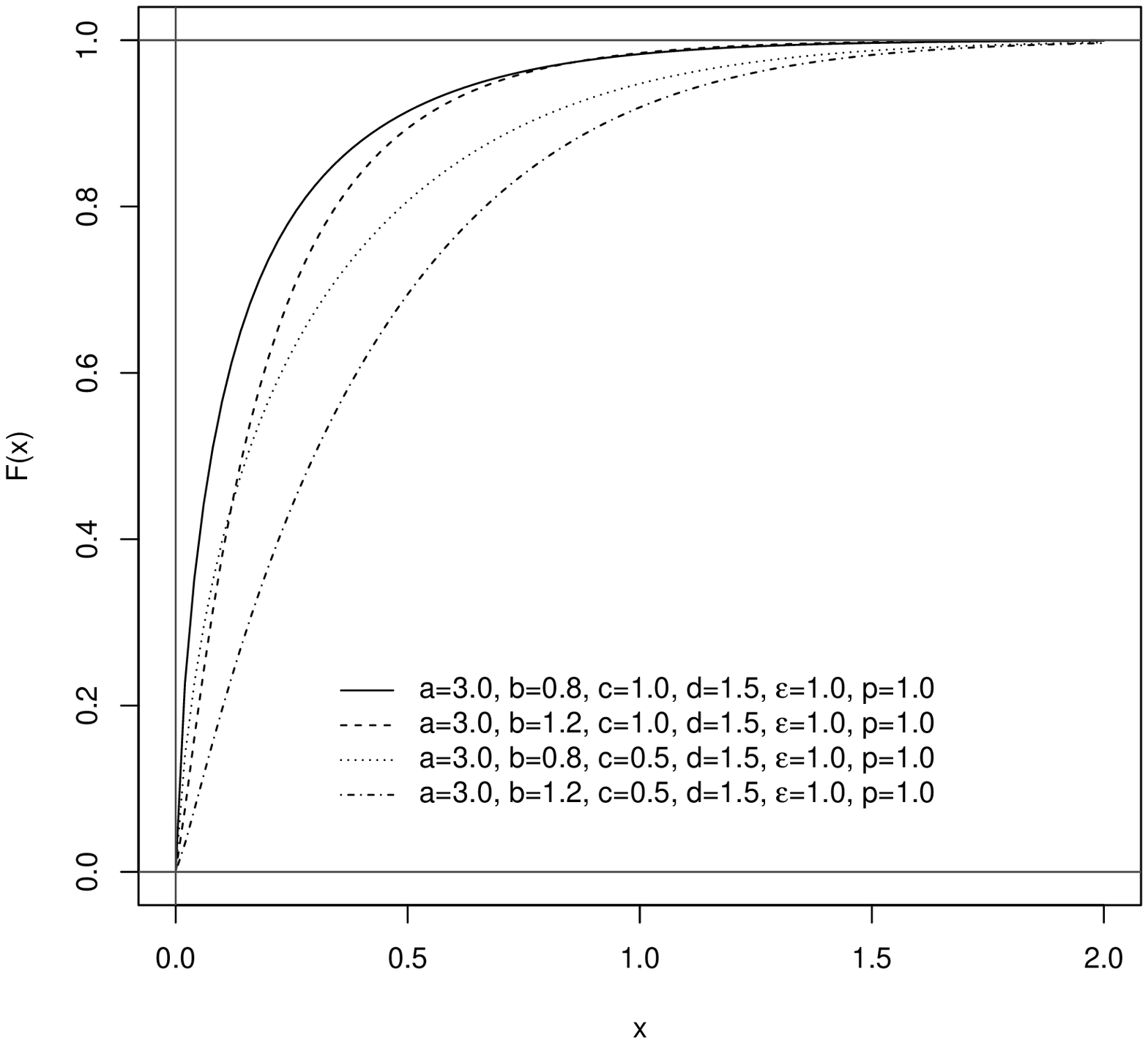} 
\label{fig01:subfig5}
}
\subfigure[$g_3$]{
\includegraphics[scale=0.32]{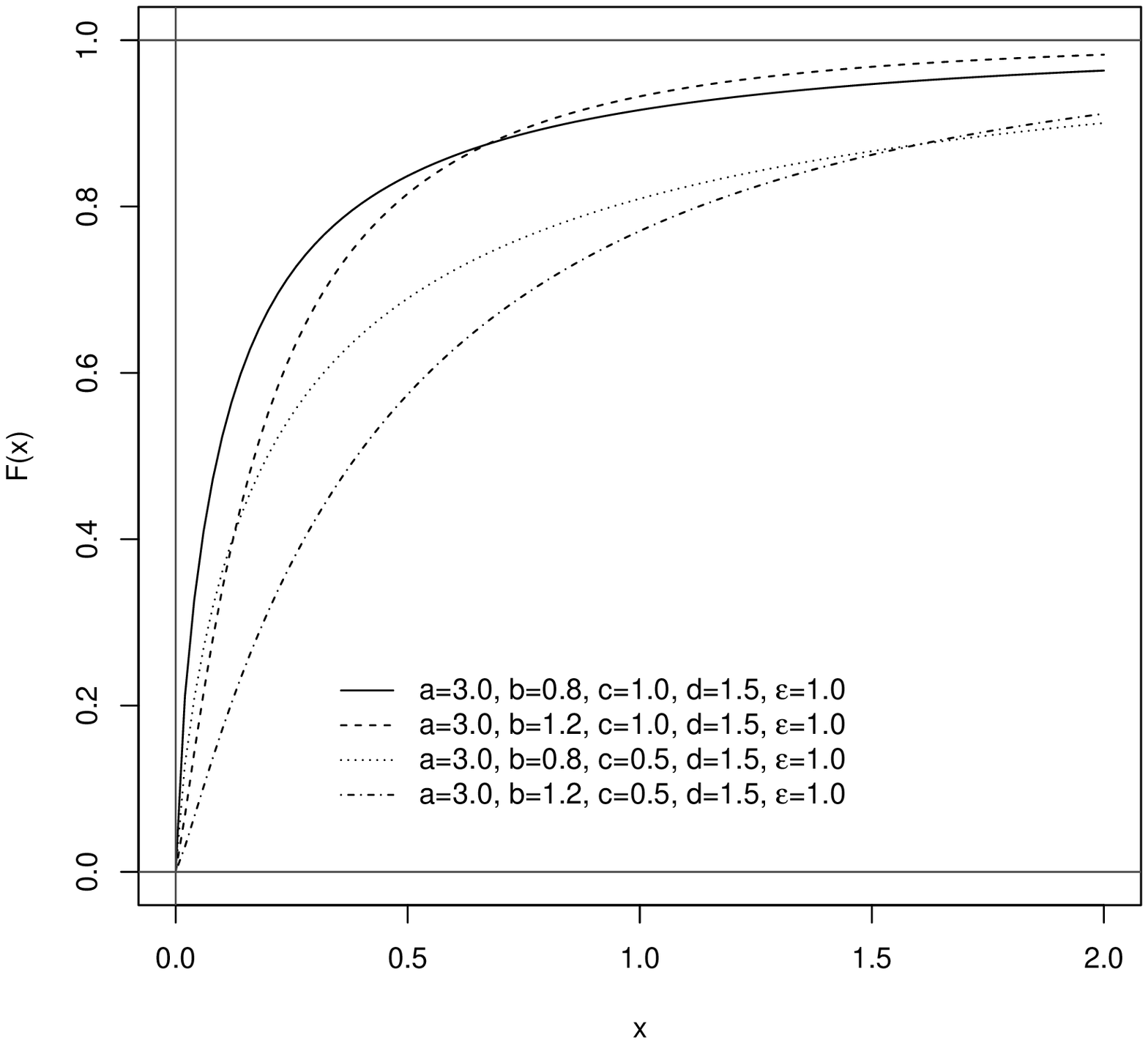} 
\label{fig01:subfig6}
}
\caption{Comparisons of cdfs associated to members of the EBXIID family}
\label{fig01}
\end{figure}

Noticing that the tail index of functions $g$ involved in (\ref{defF}) can be computed, it is found that there is correspondence between the finiteness of the tail indexes of $\barF$ and of its associated $g$.
Indeed, since $g(x)\to\infty$ as $x\to\infty$, we have
$$
\lim_{x\to\infty}\frac{\log\barF(x)}{\log x}=
-a\times\lim_{x\to\infty}\frac{\log\big(1+g(x)\big)}{\log x}=
-a\times\lim_{x\to\infty}\frac{\log g(x)}{\log x}.
$$
Hence, the finiteness of $\barF$ can be concluded from the one of its associated function $g$.
Moreover, the tail index of $\barF$ can be computed as $a$ times the tail index of $g$.

\section{Statistical Properties of Members of the EBXIID Family}
\label{sec2}

In this section statistical properties of the members of the EBXIID family studied in this paper are analyzed.

\subsection{Mode}

Since the mode of a continuous probability distribution is the value at which its probability density function has its maximum value,
the mode of $F$ given in (\ref{defF}) is then the solution $x_m$ of $F''(x)=0$ assuming that $g$ associated to $F$ is at least twice differentiable.

The conditions to define $x_m$ when considering the members of the EBXIID family studied in this paper are collected in the following result.

\begin{prop}\label{PropMode}
The mode $x_m$ of $F'$ with $F$ defined in (\ref{defF}) exists and either satisfies
$$
(a+1)\big[g'(x_m)\big]^2=(1+g(x_m))g''(x_m)
$$
where $g$ is defined in Table \ref{tab02}, or $x_m=\inf\textrm{support}\,(g)$.
\end{prop}

For instance, when considering $g_0$ with parameters $a$, $b$, $c$ and $d$, i.e. the BXII distribution, with support $\big(-d,\infty\big)$, then,
if $b>1$, after straightforward computations,
$$
x_m=\frac{1}{c}\left(\frac{b-1}{ab+1}\right)^{1/b}-d;
$$
and, if \ $0<b\leq1$, $x_m=-d$.

\subsection{Hazard Function}

The hazard function (also known as the failure rate, hazard rate, or force of mortality) $h(x)$ is the ratio between the probability density function $f(x)$ and the survival function $\barF(x)$, i.e.
$$
h(x):=\frac{f(x)}{\barF(x)}.
$$

When considering the members of the EBXIID family studied in this paper, this function is then expressed by 
$$
h(x)
=\frac{a\,g'(x)}{1+g(x)}
,\quad x>d.
$$
This function is always non-negative since $g$ satisfies $g'(x)\geq0$.
Moreover, its trend can vary when $x\to\infty$.
If $g$ is mainly based on polynomial expressions, then $h(x)\to0$, whereas if $g$ is mainly conformed by exponential expressions, then $h(x)$ may show different trends.
Among the functions $g$ giving the former of these behaviors for $h$ are $g_0$, $g_1$, $g_2$ with $p=0$, 
and $g_3$,
and for the latter is 
$g_2$ if $p>0$.

As examples, let us illustrate the behaviors of $h(x)$ for $g_3$. 
We have
\begin{eqnarray*}
\lim_{x\to\infty}h(x) & = & \lim_{x\to\infty}\frac{ab\big(c\,\frac{x^\epsilon\log(x+d+1)}{\log(x+d)}\big)^b\left(\frac{\epsilon}{x}+
\frac{1}{(x+d+1)\log(x+d+1)}-\frac{1}{(x+d)\log(x+d)}\right)}{1+\big(c\,\frac{x^\epsilon\log(x+d+1)}{\log(x+d)}\big)^b} \\
 & = & ab\times\lim_{y\to\infty}\frac{y}{1+y}\times\lim_{x\to\infty}\left(\frac{\epsilon}{x}+
\frac{1}{(x+d+1)\log(x+d+1)}-\frac{1}{(x+d)\log(x+d)}\right)
 \\
 & = & 0
\end{eqnarray*}
where $y=\big(c\,\frac{x^\epsilon\log(x+d+1)}{\log(x+d)}\big)^b$. 

These behaviors of $h$ are collected in Table \ref{tab02newX}.
\begin{table}[!h]
\centering
\begin{tabular}{ccc}
\hline
\multicolumn{1}{c}{$g$} & $h(x)$ as $x\to\infty$ & Condition \\
\hline
$g_0$ & 0 \\
$g_1$ & 0 \\
$g_2$ & 0 & $p<1$ \\
 & $a$ & $p=1$ \\
 & $\infty$ & $p>1$ \\
$g_3$ & 0 \\
\hline
\end{tabular}
\caption{Trend of $h(x)$ as $x\to\infty$, with $g$ defined in Table \ref{tab01}}
\label{tab02newX}
\end{table}

Figure \ref{fig01Y} shows behaviors of $h$ for some values of its parameters when $g$ is defined in Table \ref{tab01}.
As expected, convergences of this function to 0 as $x\to\infty$ are observed for $g_0$, $g_1$ and $g_3$, whereas 
to 3.0 
for $g_2$ fixed $p=1$.

\begin{figure}
\centering 
\subfigure[$h_0$]{
\includegraphics[scale=0.32]{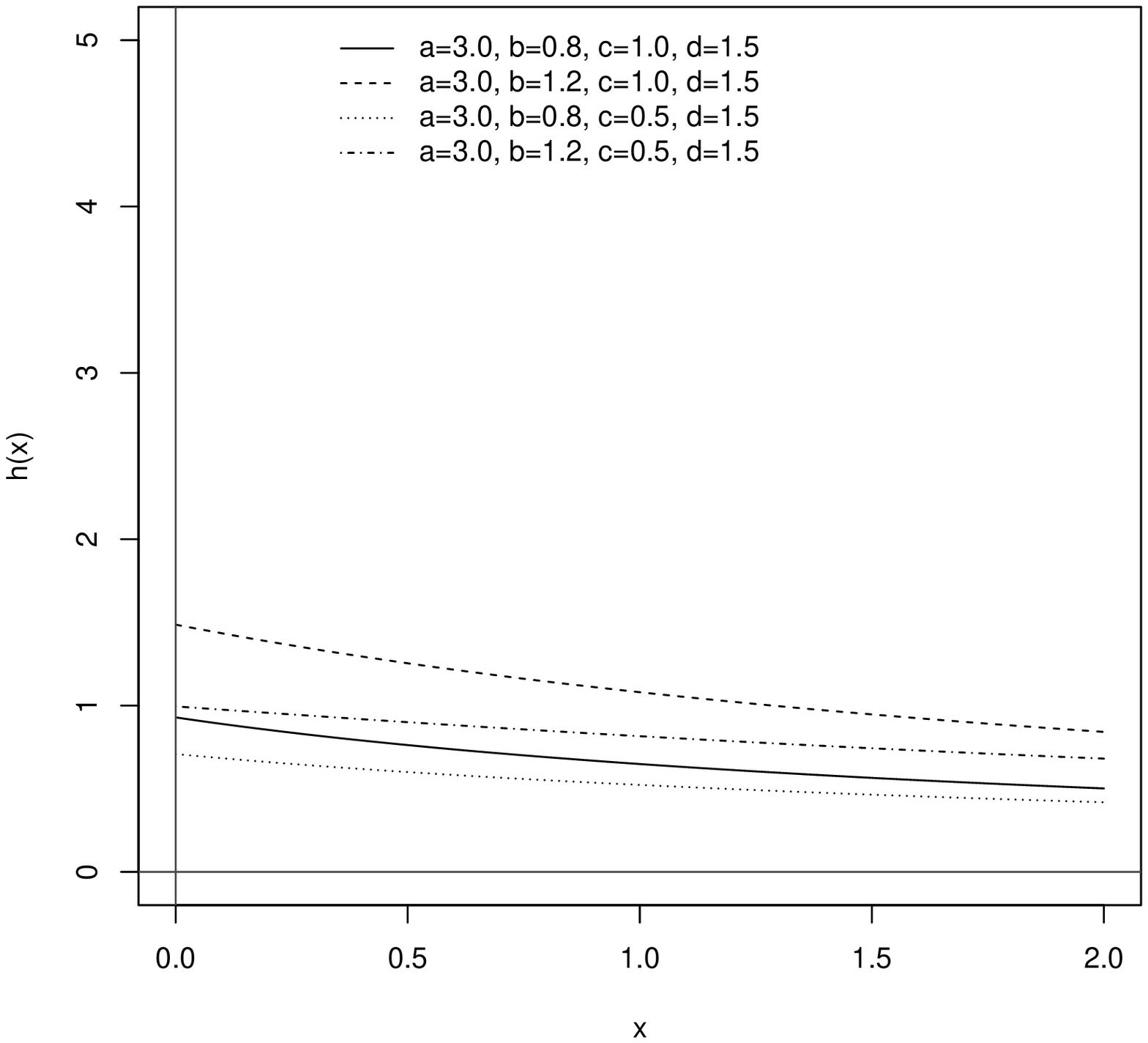} 
\label{fig01:subfig0X}
}
\subfigure[$h_1$]{
\includegraphics[scale=0.32]{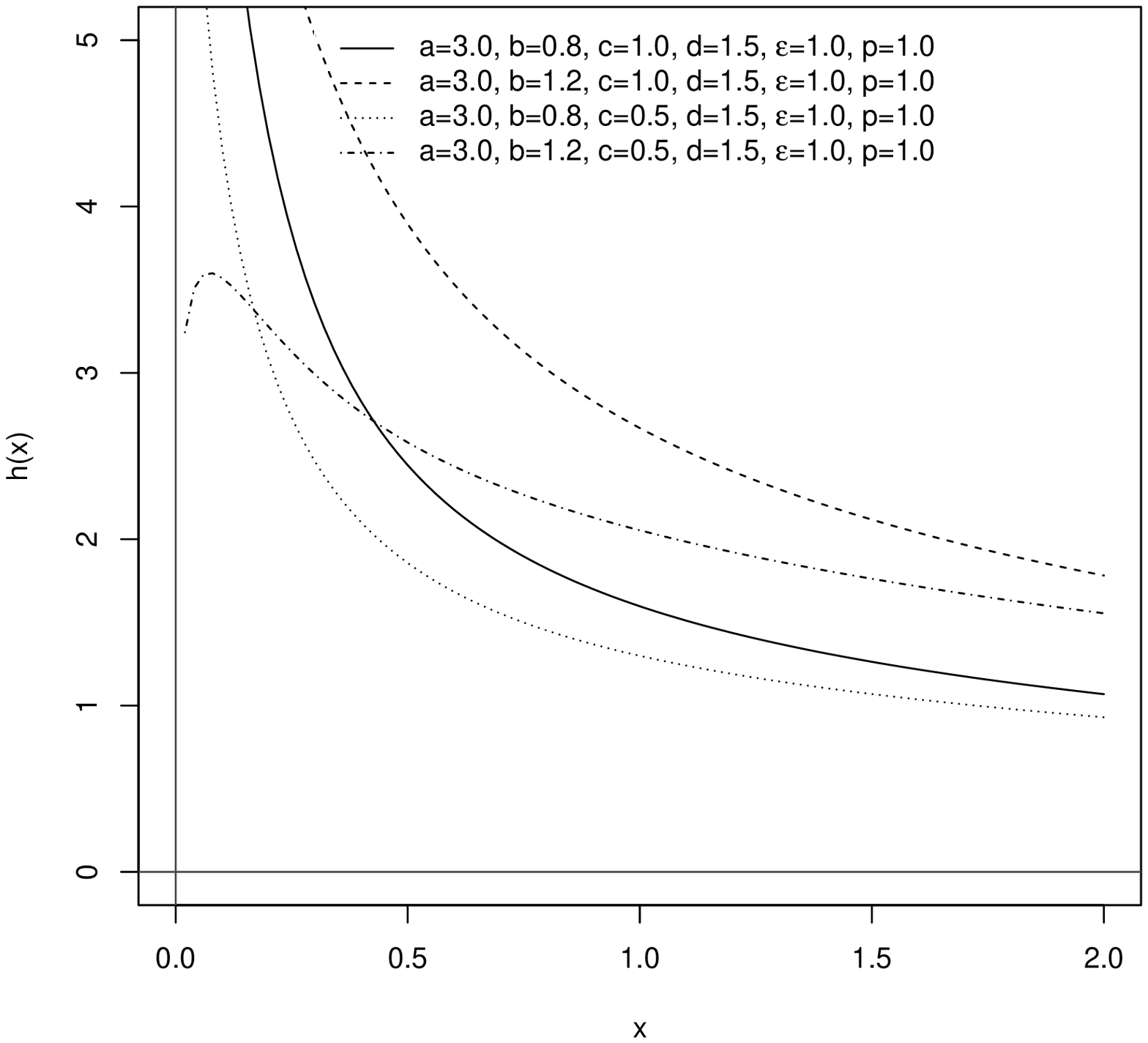} 
\label{fig01:subfig1X}
}
\subfigure[$h_2$]{
\includegraphics[scale=0.32]{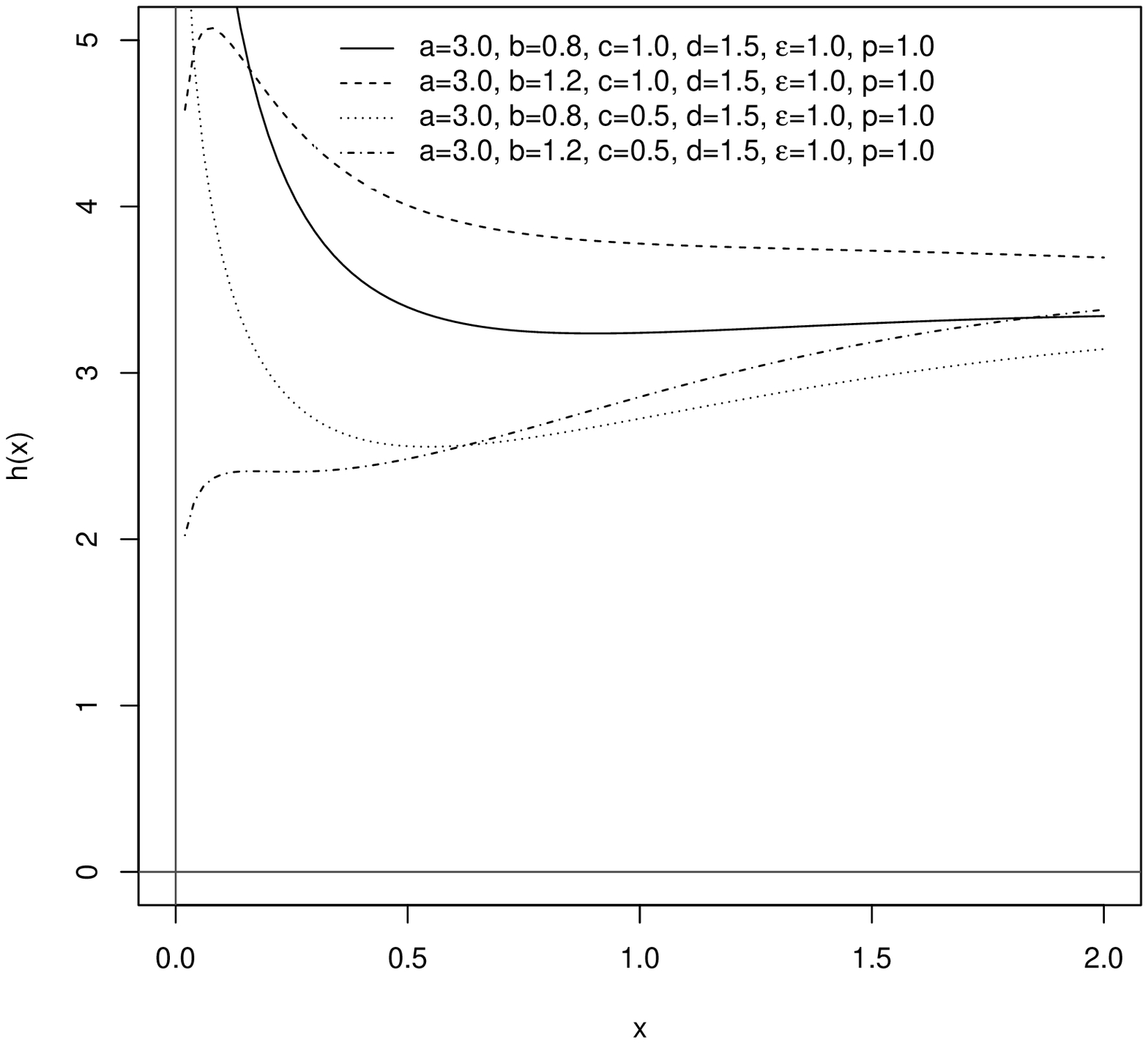} 
\label{fig01:subfig5X}
}
\subfigure[$h_3$]{
\includegraphics[scale=0.32]{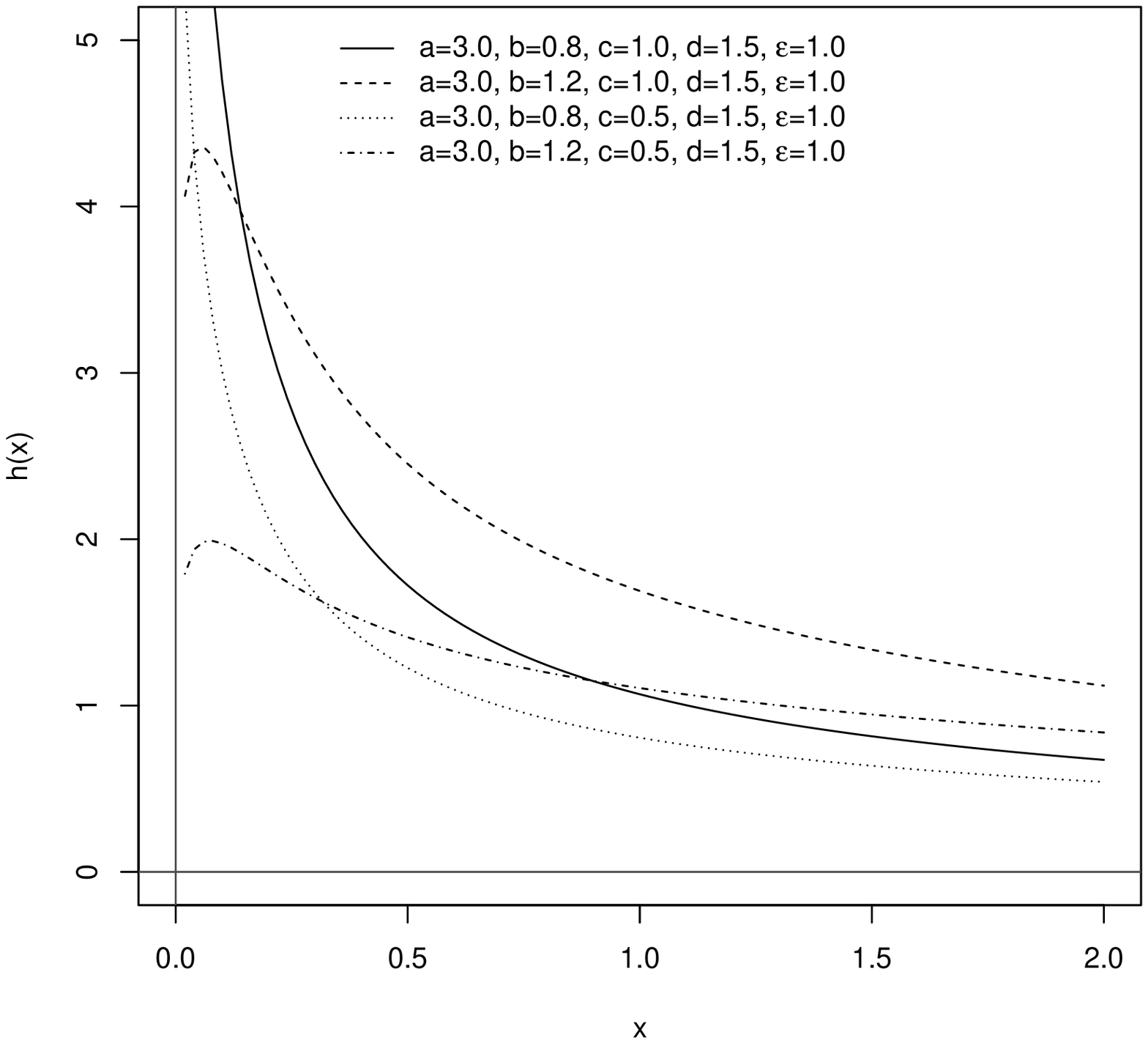} 
\label{fig01:subfig6X}
}
\caption{Comparisons of hrs associated to members of the EBXIID family}
\label{fig01Y} 
\end{figure}

\subsection{Quantiles and Random Number Generation}
\label{Quantilesandrandomnumbergeneration}

In this subsection $F$ defined by (\ref{defF}) is assumed to be strictly increasing.

Applying the inverse transform method, the quantile function $q(v)$, $0<v<1$, is then obtained by solving
\begin{equation}\label{quantile}
F\big(q(v)\big)=v.
\end{equation}
Thus, the previous relation can be applied when $g$ attached to $F$ is taken from Table \ref{tab01} since all these functions are strictly increasing.

Table \ref{exQuantile} presents illustrations of quantiles of $F$ when $g$ is taken from Table \ref{tab01}.
Computations of these values have been made using the function \texttt{uniroot} in the R software package.
Keeping fixed a quantile, this table reveals the effect of the tail index.
Lower tail index higher quantiles.
Among all parameters associated to the distributions considered, the ones related to $g_0$ and $g_3$ are the lowest, and show the highest quantiles.
These outputs are followed by the quantiles of the distributions with $g$ as $g_1$.
As expected, the quantiles for $g_2$ with $p=1$ are the lowest.
Lower differences between quantiles are observed when $v$ increases, and they increase if $v$ decreases.
This last effect is important in areas where underlying risk can become important, for instance in insurance, finance, natural disasters \cite{embrechts1997,CohenSiegelman2010,BataybalBeladi2001}.
Small variations in estimates of the tail index may produce large variations in quantiles.

\begin{table}[!h]
\centering
\begin{tabular}{clrrr}
\hline
$g$ & \multicolumn{1}{c}{Parameters} & $q(0.90)$ & $q(0.95)$ & $q(0.99)$ \\
\hline
$g_0$ & $a=3.0$, $b=0.8$, $c=1.0$, $d=1.5$ & -0.30 & 0.46 & 3.53 \\
 & $a=3.0$, $b=1.2$, $c=1.0$, $d=1.5$ & -0.37 & 0.06 & 1.43 \\
 & $a=3.0$, $b=0.8$, $c=0.5$, $d=1.5$ & 0.89 & 2.42 & 8.56 \\
 & $a=3.0$, $b=1.2$, $c=0.5$, $d=1.5$ & 0.75 & 1.63 & 4.37 \\
$g_1$ & $a=3.0$, $b=0.8$, $c=1.0$, $d=1.5$, $\epsilon=1.0$, $p=1.0$ & 0.40 & 0.70 & 1.81 \\
 & $a=3.0$, $b=1.2$, $c=1.0$, $d=1.5$, $\epsilon=1.0$, $p=1.0$ & 0.37 & 0.54 & 1.07 \\
 & $a=3.0$, $b=0.8$, $c=0.5$, $d=1.5$, $\epsilon=1.0$, $p=1.0$ & 0.87 & 1.43 & 3.29 \\
 & $a=3.0$, $b=1.2$, $c=0.5$, $d=1.5$, $\epsilon=1.0$, $p=1.0$ & 0.81 & 1.15 & 2.09 \\
$g_2$ & $a=3.0$, $b=0.8$, $c=1.0$, $d=1.5$, $\epsilon=1.0$, $p=1.0$ & 0.45 & 0.66 & 1.15 \\
 & $a=3.0$, $b=1.2$, $c=1.0$, $d=1.5$, $\epsilon=1.0$, $p=1.0$ & 0.51 & 0.69 & 1.11 \\
 & $a=3.0$, $b=0.8$, $c=0.5$, $d=1.5$, $\epsilon=1.0$, $p=1.0$ & 0.75 & 1.01 & 1.57 \\
 & $a=3.0$, $b=1.2$, $c=0.5$, $d=1.5$, $\epsilon=1.0$, $p=1.0$ & 0.92 & 1.16 & 1.67 \\
$g_3$ & $a=3.0$, $b=0.8$, $c=1.0$, $d=1.5$, $\epsilon=1.0$ & 0.84 & 1.56 & 4.64 \\
 & $a=3.0$, $b=1.2$, $c=1.0$, $d=1.5$, $\epsilon=1.0$ & 0.78 & 1.18 & 2.53 \\
 & $a=3.0$, $b=0.8$, $c=0.5$, $d=1.5$, $\epsilon=1.0$ & 1.99 & 3.52 & 9.71 \\
 & $a=3.0$, $b=1.2$, $c=0.5$, $d=1.5$, $\epsilon=1.0$ & 1.85 & 2.73 & 5.49 \\
\hline
\end{tabular}
\caption{Quantiles for the distributions shown in Figure \ref{fig01}}
\label{exQuantile}
\end{table}

The solution $q$ of (\ref{quantile}) given $v$, $0<v<1$, could be used to generate random numbers of a random variable (rv) that follows members of the EBXIID family assuming that $F'>0$.
This is an application of the inverse transform sampling method to draw random samples.
This method to generate random numbers consists in \cite{Devroye}:
\begin{enumerate}
\item
Generate a random number $v$ from the standard uniform distribution in the interval $[0,1]$; and,

\item
Compute $q$ such that $F(q)=v$, i.e. $q$ satisfies (\ref{quantile}).

\end{enumerate}

The implementation of the previous method may be done by generating random numbers following a uniform distribution that may be performed using the function \texttt{runif} in the R software package, and, after, by computing quantiles that may be performed using the function \texttt{uniroot} mentioned above.

We will come back on this random number generation procedure later in order to simulate random numbers following members of the EBXIID family.
These numbers will be used to study the performance of the new distributions introduced in this paper.

\section{Maximum Likelihood Estimation}
\label{sec3}

In this section we propose 
the method of maximum likelihood for estimating the parameters of the members of the EBXIID family with $g$ defined in Table \ref{tab01}.

Let $X$ be a rv following one of such members, say a twice differentiable cdf $F$, with parameters $\theta=(\theta_1,\ldots,\theta_w)$, and let $x_1$, \ldots, $x_n$ be a sample of $X$ obtained independently.

Following the method of maximum likelihood, the likelihood function of this random sample is then given by
$$
L(\theta|x_1,\ldots,x_n)=\prod_{k=1}^nF'(\theta,x_k),
$$
and its log-likelihood function then is
$$
l(\theta|x_1,\ldots,x_n)=\sum_{k=1}^n\log F'(\theta,x_k).
$$
Maximum likelihood estimates (MLEs) of $\theta_1$, \ldots, $\theta_w$ might be reached by solving the non-linear system obtained by equaling to 0 the derivatives of $l$ with respect to $\theta$, i.e., for $i$ = 1, \ldots, $w$,
\begin{eqnarray*}
\frac{\partial l}{\partial\theta_i} & = & \sum_{k=1}^n\frac{1}{F'(\theta,x_k)}\frac{\partial F'(\theta,x_k)}{\partial\theta_i}\quad=\quad0.
\end{eqnarray*}
There are not explicit solutions for this system.
A method to numerically solve such a system is the Newton-Raphson (NR) algorithm.
This is a well-known and useful technique for finding roots of systems of non-linear equations in several variables. 
We use the function \texttt{optimx} in the R software package, where by choosing the optimization method called limited-memory BFGS \cite{ByrdNocedalZhu1995} (L-BFGS-B) that is a type of quasi-Newton optimization and allows box-constrained optimization, a minimization of an objective function can be carried out.
In our case the function \texttt{optimx} is applied to the objective function $-l(\theta|x_1,\ldots,x_n)$ in order to obtain MLEs $\hat{\theta}$ of $\theta$.

For interval estimation of $\theta$ and hypothesis tests on these parameters, we use the $w\times w$ observed information matrix given by
$$
I(\theta)
=-E\left[
\begin{array}{ccc}
\displaystyle \frac{\partial^2 l}{\partial\theta_1^2} & \cdots & \displaystyle \frac{\partial^2 l}{\partial\theta_1\partial\theta_w} \\
 \vspace{-3mm} & \\
\vdots & \ddots & \vdots \\
 \vspace{-3mm} & \\
\displaystyle \frac{\partial^2 l}{\partial\theta_w\partial\theta_1} & \cdots & \displaystyle \frac{\partial^2 l}{\partial\theta_w^2}
\end{array}
\right]
$$
where
$$
\frac{\partial^2 l}{\partial\theta_i\partial\theta_j}=\sum_{k=1}^n\left(\frac{1}{F'(\theta,x_k)}\frac{\partial^2 F'(\theta,x_k)}{\partial\theta_i\partial\theta_j}-\frac{1}{\big[F'(\theta,x_k)\big]^2}\frac{\partial F'(\theta,x_k)}{\partial\theta_i}
\frac{\partial F'(\theta,x_k)}{\partial\theta_j}\right).
$$
Under certain regularity conditions \cite{Cramer}, the MLE $\hat{\theta}$ approximates as $n$ increases a multivariate normal distribution with
mean equal to the true parameter value $\theta$ and variance-covariance matrix given
by the inverse of the observed information matrix, i.e. $\Sigma=\big[\sigma_{ij}\big]=I^{-1}(\theta)$.
Hence, the asymptotic behavior of two-sided $(1-\tau)\times100~\%$ confidence intervals for the
parameter $\theta_1$, \ldots, $\theta_w$ are approximately, for $i$ = 1, \ldots, $w$,
$$
\hat{\theta_i}\pm z_{\epsilon/2}\sqrt{\hat{\sigma}_{ii}}
$$
where $z_\delta$ represents the $\delta\times100~\%$ percentile of the standard normal distribution.

Expressions for estimating the parameters of the members of the EBXIID family with $g$ as in Table \ref{tab01} via the method of maximum likelihood and for estimating their information matrices are presented in Annexe \ref{RelationsForEstimatingParameters}.

\section{Simulation Studies}
\label{sec4}

In this section, Monte Carlo simulations to assess the performance of the MLEs of the members of the EBXIID family with $g$ as in Table \ref{tab01} are performed.
The parameter values shown in Figure \ref{fig01} are adopted as the true parameters,
where he following parameters are always fixed: $a=3.0$, $d=1.5$, $\epsilon=1.0$ and $p=1.0$.

Simulation scenarios are produced by varying the sample size $n$: 1\,000 and 10\,000.
These simulations are based on the generation of random numbers as indicated in Subection \ref{Quantilesandrandomnumbergeneration}.
In order to always replicate the same results a fixed seed has been used.
The code used in these studies is available upon request.

Table \ref{ParameterestimatesSimulation} shows estimates for $b$ and $c$ and log-likelihoods.
Standard errors in brackets.
Estimates and standard errors obtained from observed Hessian matrices have been computed using the function \texttt{optimx} mentioned above.
These results show that the MLE method works well since estimates approach nicely the true parameters.
Further, as expected, the errors of these estimates decrease when $n$ increases.

\begin{table}[!h]
\centering
{\small
\begin{tabular}{crccccccr}
\hline
 & & & \multicolumn{2}{c}{Parameters$^{(*)}$} & & \multicolumn{2}{c}{Estimates} & \multicolumn{1}{c}{Log-likelihood} \\
\cline{4-5}
\cline{7-9}
$g$ & \multicolumn{1}{c}{$n$} & & $b$ & $c$ & & $b$ & $c$ & \multicolumn{1}{c}{$-l$} \\
\hline
$g_1$ & 1\,000 & & 0.8 & 1.0 & & $0.781_{(0.019)}$ & $0.963_{(0.056)}$ & $-1099$ \\
      & 10\,000 & & 0.8 & 1.0 & & $0.788_{(0.006)}$ & $0.976_{(0.018)}$ & $-10920$ \\
\cline{2-9}
      & 1\,000 & & 1.2 & 1.0 & & $1.189_{(0.030)}$ & $0.984_{(0.037)}$ & $-833$ \\
      & 10\,000 & & 1.2 & 1.0 & & $1.194_{(0.009)}$ & $0.991_{(0.012)}$ & $-8309$ \\
\cline{2-9}
      & 1\,000 & & 0.8 & 0.5 & & $0.791_{(0.020)}$ & $0.487_{(0.028)}$ & $-360$ \\
      & 10\,000 & & 0.8 & 0.5 & & $0.793_{(0.006)}$ & $0.491_{(0.009)}$ & $-3550$ \\
\cline{2-9}
      & 1\,000 & & 1.2 & 0.5 & & $1.188_{(0.030)}$ & $0.492_{(0.018)}$ & $-78$ \\
      & 10\,000 & & 1.2 & 0.5 & & $1.194_{(0.009)}$ & $0.495_{(0.006)}$ & $-758$ \\
\hline
$g_2$ & 1\,000 & & 0.8 & 1.0 & & $0.779_{(0.023)}$ & $0.952_{(0.063)}$ & $-891$ \\
      & 10\,000 & & 0.8 & 1.0 & & $0.790_{(0.007)}$ & $0.975_{(0.020)}$ & $-8860$ \\
\cline{2-9}
      & 1\,000 & & 1.2 & 1.0 & & $1.176_{(0.036)}$ & $0.972_{(0.044)}$ & $-513$ \\
      & 10\,000 & & 1.2 & 1.0 & & $1.194_{(0.011)}$ & $0.990_{(0.014)}$ & $-5119$ \\
\cline{2-9}
      & 1\,000 & & 0.8 & 0.5 & & $0.778_{(0.026)}$ & $0.472_{(0.035)}$ & $-313$ \\
      & 10\,000 & & 0.8 & 0.5 & & $0.795_{(0.008)}$ & $0.491_{(0.011)}$ & $-3047$ \\
\cline{2-9}
      & 1\,000 & & 1.2 & 0.5 & & $1.170_{(0.041)}$ & $0.481_{(0.026)}$ & $-102$ \\
      & 10\,000 & & 1.2 & 0.5 & & $1.194_{(0.013)}$ & $0.494_{(0.008)}$ & $-1066$ \\
\hline
$g_3$ & 1\,000 & & 0.8 & 1.0 & & $0.789_{(0.020)}$ & $0.972_{(0.056)}$ & $-455$ \\
      & 10\,000 & & 0.8 & 1.0 & & $0.792_{(0.006)}$ & $0.981_{(0.018)}$ & $-4494$ \\
\cline{2-9}
      & 1\,000 & & 1.2 & 1.0 & & $1.188_{(0.030)}$ & $0.984_{(0.037)}$ & $-161$ \\
      & 10\,000 & & 1.2 & 1.0 & & $1.194_{(0.009)}$ & $0.991_{(0.012)}$ & $-1591$ \\
\cline{2-9}
      & 1\,000 & & 0.8 & 0.5 & & $0.791_{(0.020)}$ & $0.487_{(0.028)}$ & $-350$ \\
      & 10\,000 & & 0.8 & 0.5 & & $0.796_{(0.006)}$ & $0.493_{(0.009)}$ & $-3590$ \\
\cline{2-9}
      & 1\,000 & & 1.2 & 0.5 & & $1.188_{(0.030)}$ & $0.492_{(0.018)}$ & $-678$ \\
      & 10\,000 & & 1.2 & 0.5 & & $1.194_{(0.009)}$ & $0.495_{(0.006)}$ & $-6823$ \\
\hline
\multicolumn{9}{l}{$^{(*)}$ Fixed parameters: $a=3.0$, $d=1.5$, $\epsilon=1.0$ and $p=1.0$}
\end{tabular}
}
\caption{Parameter estimates for selected models}
\label{ParameterestimatesSimulation}
\end{table}

\section{Applications}
\label{sec5}

In this section, we present applications in order to illustrate the
performance and usefulness of the proposed distribution family when compared to natural competitors.

To this aim, real data from several domains are used.
In all cases these data have been analyzed in other researches and, in this paper, are fitted using members of the EBXIID family.
This allows the immediate comparison of our results with respect to the ones of competitors.
These comparisons are done using the following two well-known measures of the relative quality of statistical models for given data: the Akaike information criterion (AIC),
defined by $2m -2l$ where $m$ is the number of parameters
of the model, and the Bayesian information criterion (BIC), also called the Schwarz information criterion, 
defined by $m \log n -2l$ with $n$ the sample size.
For both criteria, the lower the better.

Through all applications the BXII distribution, given by $g=g_0$, is always included in order to compare other members of the EBXIID family with respect to this distribution.
The parameters of the members of the EBXIID family are always estimated using the procedure of maximum likelihood described in Section \ref{sec3}.

Parameter estimates and standard errors are computed using the function \texttt{optimx}.

\subsection{Strengths Data}

In the first application, the sample I of strengths reported by Smith and Naylor in \cite{SmithNaylor1987b} is modeled.
These authors analyzed these data using the Weibull distribution.
These data consisted in samples of experimental data of the strength of glass fibres of length 1.5 cm, from the National Physical Laboratory in England.
In this sample, Smith and Naylor did not give the unit of measurement.

The procedure of maximum likelihood described in Section \ref{sec3} for estimating the parameters of members of the EBXIID family is applied.
Following this procedure, when considering the analyzed fatigue data, the maximum likelihood estimates presented in 
Table \ref{FatigueGELS} are obtained.
Standard errors computed from the observed Hessian matrix in brackets.
Note that some parameters have been fixed \emph{a priori} or have been related each other in order to retain simple models.

\begin{table}[!h]
\centering
\begin{tabular}{ccccccc}
\hline
$g$ & $a$ & $b$ & $c$ & $d$ & $\epsilon$ & $p$ \\
\hline
$g_0$ & $4.040_{(3.339)}$ & $10.641_{(1.760)}$ & $0.346_{(0.036)}$ & 1 (fixed) & \\
$g_1$ $^{(1)}$ & $2.728_{(0.142)}$ & $2.728_{(0.142)}$ & $0.282_{(0.027)}$ & 2 (fixed) & $2.728_{(0.142)}$ & 0 (fixed) \\
$g_2$ $^{(2)}$ & $1.096_{(0.538)}$ & $2.591_{(0.306)}$ & $0.273_{(0.073)}$ & 2 (fixed) & 1 (fixed) & $2.591_{(0.306)}$ \\
$g_3$ $^{(3)}$ & $8.469_{(0.866)}$ & $8.469_{(0.866)}$ & $0.370_{(0.007)}$ & 1 (fixed) & 1 (fixed) \\
\hline
\multicolumn{7}{l}{$^{(1)}$ Taken $a=b=\epsilon$} \\
\multicolumn{7}{l}{$^{(2)}$ Taken $b=p$} \\
\multicolumn{7}{l}{$^{(3)}$ Taken $a=b$}
\end{tabular}
\caption{Fit of strengths data using members of the EBXIID family}
\label{FatigueGELS}
\end{table}

The fatigue data examined are popular. They have been studied by several authors.
For instance, Smith and Naylor \cite{SmithNaylor1987b} fitted them using a Weibull model.
Next, skew $t$-distributions have been used to fit these data as done by Jones and Faddy \cite{JonesFaddy2003} and more recently Baker \cite{Baker2016} with his generalized asymmetric (GAT) distribution.
Also, Morais and Barreto-Souza \cite{MoraisBarretoSouza2011} examined these data using the Weibull Poisson (WP), Rayleigh Poisson (RP), and exponential Poisson (EP) distributions.
On the other hand, Barreto et al. \cite{BarretoSouzaSantosCordeiro2010} applied to these data the beta generalized exponential (BGE) and the beta exponential (BE) distributions; Jones and Pewsey \cite{JonesPewsey2009} modeled them using the sinh-arcsinh (SHASH) distribution and its normal (SHASH-N), normal-tailed (SHASH-NT) and symmetric (SHASH-S) sub-models; and,
Ma and Genton \cite{MaGenton2004} fitted them using flexible generalized skew-normal (FGSN) distributions by varying a parameter $K$.
Table \ref{FatigueTab} collects the AIC and BIC values reported by the above-cited authors on the models mentioned above.
In such a table the AIC and BIC values of the new models indicated in Table \ref{FatigueGELS} are included.
The number of model parameters, $w$, and log-likelihoods are also included in that table.
These results show that the decision for choosing a model varies according to the parsimony criterion one. 
On the one hand, the AIC favors the SHASH distribution, whereas, on the other hand, the BIC the member of the EBXIID family with $g=g_3$.
All these models are highlighted.
The preference for this last model reveals that data would be heavy-tailed, but with a high shape parameter given by the tail index, more than 64.
This puts in evidence why light-tailed models like $g_2$ that presents the second lowest for both AIC and BIC values may also give good fits.
From these results, it is noted that the performance of the BXII distribution that is a good model alternative can be improved by some members of the EBXIID family.

\begin{table}[!h]
\centering
\begin{tabular}{lcccc}
\hline
\multicolumn{1}{c}{Model} & $w$ & $-l$ & AIC & BIC \\
\hline
BE & 4 & 24.12 & 54.24 & 60.03 \\
BGE & 4 & 15.59 & 39.18 & 46.90 \\
EBXIID family ($g_0$) & 3 & 12.70 & 31.41 & 37.20 \\
EBXIID family ($g_1$) & 2 & 14.57 & 33.14 & 37.00 \\
EBXIID family ($g_2$) & 3 & 11.22 & 28.44 & 34.24 \\
EBXIID family ($g_3$) & 2 & 12.95 & 29.91 & $\mathbf{33.78}$ \\
EP & 2 & 30.51 & 65.02 & 68.88 \\
FGSN ($K=1$) & 4 & 11.95 & 31.90 & 40.40 \\
FGSN ($K=3$) & 5 & 11.60 & 33.20 & 43.90 \\
GAT & 4 & 11.75 & 31.50 & 39.22 \\
RP & 2 & 20.81 & 45.62 & 49.48 \\
SHASH & 4 & 10.00 & $\mathbf{28.00}$ & 36.57 \\
SHASH-N & 2 & 17.92 & 39.84 & 44.13 \\
SHASH-NT & 3 & 13.63 & 33.26 & 39.69 \\
SHASH-S & 3 & 13.46 & 32.92 & 39.35 \\
Weibull & 3 & 14.55 & 35.11 & 40.90 \\
WP & 3 & 13.47 & 32.94 & 38.73 \\
\hline
\end{tabular}
\caption{Fatigue data: log-likelihoods and AIC and BIC values}
\label{FatigueTab}
\end{table}

\subsection{Rainfall Data}

In the second application, rainfall data reported by Chen et al. \cite{ChenBunceJiang2010} are examined.
These observations measured in mm are annual maximum antecedent rainfalls of a 60-day duration taken from Maple Ridge in British Columbia,
Canada.

Following the maximum likelihood method described in Section \ref{sec3} for estimating the parameters of members of the EBXIID family, for the studied rainfall data, the maximum likelihood estimates presented in 
Table \ref{RainfallGELS} are obtained.
Standard errors computed from the observed Hessian matrix in brackets.
Some of these models have been kept simple by fixing parameters.

\begin{table}[!h]
\centering
\begin{tabular}{ccccccc}
\hline
$g$ & $a$ & $b$ & $c$ & $d$ & $\epsilon$ & $p$ \\
\hline
$g_0$ & $4.719_{(1.337)}$ & $5.346_{(0.823)}$ & $0.001_{(0.0001)}$ & 1 (fixed) & \\
$g_1$ & $5.736_{(6.802)}$ & $6.964_{(1.158)}$ & $0.007_{(0.001)}$ & 1 (fixed) & 1 (fixed) & 0 (fixed) \\
$g_2$ & $6.776_{(0.688)}$ & $6.776_{(0.688)}$ & $0.006_{(0.0002)}$ & $1$ (fixed) & 1 (fixed) & $0$ (fixed) \\
$g_3$ $^{(1)}$ & $16.658_{(1.653)}$ & $16.658_{(1.653)}$ & $0.100_{(0.086)}$ & $0.100_{(0.086)}$ & $0.330_{(0.001)}$ \\
\hline
\multicolumn{7}{l}{$^{(1)}$ Taken $a=b$ and $c=d$}
\end{tabular}
\caption{Fit of rainfall data using members of the EBXIID family}
\label{RainfallGELS}
\end{table}

The rainfall data considered in this analysis were analyzed by Tahir et al. \cite{TahirCordeiroAlzaatrehMansoorZubair2014} by using the logistic-Fr\'{e}chet (LFr) distribution introduced by these authors.
They also fitted those data to the Marshall-Olkin Fr\'{e}chet (MOFr) distribution introduced by Krishna et al. \cite{KrishnaJoseAliceRistic2013} and Krishna et al. \cite{KrishnaJoseRistic2013}, the exponentiated-Fr\'{e}chet (EFr) introduced by Nadarajah and Kotz \cite{NadarajahKotz2003}, and the Fr\'{e}chet (Fr) distribution.

The AIC and BIC values of all these models reported by the above-cited authors, AIC and BIC values of the members of the EBXIID family analyzed in this paper, and their corresponding log-likelihoods and parameter numbers, $w$, are collected in Table \ref{RainfallTab}.
Among all these models, both the AIC and BIC values favor the ones of the EBXIID family, and among them the distribution with $g=g_2$.
These results seem to show that data are heavy-tailed, but with a high shape parameter, more than 40.
Note that in this case the distributions with $g$ as $g_1$, $g_2$ or $g_3$, all of them show better performances than the BXII distribution.

\begin{table}[!h]
\centering
\begin{tabular}{lcccc}
\hline
\multicolumn{1}{c}{Model} & $w$ & $-l$ & AIC & BIC \\
\hline
EBXIID family ($g_0$) & 3 & 326.68 & 659.37 & 665.23 \\
EBXIID family ($g_1$) & 3 & 324.97 & 655.95 & 661.81 \\
EBXIID family ($g_2$) & 2 & 325.03 & $\mathbf{654.06}$ & $\mathbf{657.97}$ \\
EBXIID family ($g_3$) & 3 & 325.17 & 656.35 & 662.20 \\
EFr & 3 & 328.11 & 662.23 & 668.09 \\
Fr & 2 & 341.15 & 686.31 & 690.21 \\
LFr & 3 & 327.19 & 660.38 & 666.23 \\
MOFr & 3 & 327.22 & 660.45 & 666.30 \\
\hline
\end{tabular}
\caption{Rainfall data: log-likelihoods and AIC and BIC values}
\label{RainfallTab}
\end{table}

\subsection{Roller Data}

In the third application, data on surface roughness of rollers reported by Laslett \cite{Laslett1994} are examined.
They correspond to the second subset of data used in \cite{Laslett1994}.
The data are available for downloading at \url{http://lib.stat.cmu.edu/jasadata/laslett}.
In that site the following background information is provided: ``The dataset consists of 1150 heights measured at 1 micron intervals along
 the drum of a roller (i.e. parallel to the axis of the roller). 
This was part of an extensive study of surface roughness of the rollers. The units of height are not 
 given, because the data are automatically rescaled as they are recorded, and the
 scaling factor is imperfectly known. The zero reference height is arbitrary.''

Applying the maximum likelihood method explained in Section \ref{sec3} for estimating the parameters of members of the EBXIID family, for the studied rainfall data, the obtained estimates are presented in 
Table \ref{RollerGELS}.
Standard errors computed from the observed Hessian matrix in brackets.
Parameters in some models have been fixed in order to keep them simple.

\begin{table}[!h]
\centering
\begin{tabular}{ccccccc}
\hline
$g$ & $a$ & $b$ & $c$ & $d$ & $\epsilon$ & $p$ \\
\hline
$g_0$ & $20.364$ & $16.339_{(1.797)}$ & $0.101_{(0.010)}$ & $4.376_{(0.892)}$ & 1066 \\
$g_1$ $^{(1)}$ & $6.606_{(1.110)}$ & $2.650_{(0.243)}$ & $0.008_{(0.007)}$ & $1.325_{(0.143)}$ & $2.650_{(0.243)}$ & $0.650_{(0.686)}$ \\
$g_2$ $^{(2)}$ & $2.639_{(0.023)}$ & $1.319_{(0.013)}$ & $0.003_{(0.0001)}$ & $1.319_{(0.013)}$ & 1 (fixed) & $2.639_{(0.023)}$ \\
$g_3$ & $18.527_{(4.901)}$ & $3.609_{(0.272)}$ & $0.032_{(0.008)}$ & $5.765$ & $1.933_{(0.145)}$ \\
\hline
\multicolumn{7}{l}{$^{(1)}$ Taken $b=2d=\epsilon$} \\
\multicolumn{7}{l}{$^{(2)}$ Taken $a=2b=2d=p$}
\end{tabular}
\caption{Fit of roller data using members of the EBXIID family}
\label{RollerGELS}
\end{table}

The roller data considered in this study have been usually treated as Gaussian, see e.g. \cite{ConstantineHall1994,KentWood1997,BlankeVial2014}.
Recently, Balakrishnan et al. \cite{BalakrishnanSauloLeao2017} analyzed those data using the skew exponential-power
Birnbaum-Saunders (SEPBS) distribution and the skew Student-t Birnbaum-Saunders (StBS) distribution that are based on the Birnbaum-Saunders (BS) distribution. As explained by these authors, the BS distribution is related to the normal distribution through a stochastic representation, but presenting a positive skew.
These authors also fitted those data using the BS, log-normal and gamma distributions.

Table \ref{RollerTab} shows both the AIC and BIC values reported by the above-cited authors on the models above mentioned and also the ones associated to the members of the EBXIID family studied in this paper.
Model parameter numbers and log-likelihoods are also included in that table.
According to those parsimony criteria, the distribution with $g=g_2$ would be the favorite model.
This result is in line with typical assumptions on the data roller where they are considered light-tailed \cite{ConstantineHall1994,KentWood1997,BlankeVial2014}.

\begin{table}[!h]
\centering
\begin{tabular}{lcccc}
\hline
\multicolumn{1}{c}{Model} & $w$ & $-l$ & AIC & BIC \\
\hline
EBXIID family ($g_0$) & 4 & 1066.96 & 2141.93 & 2162.12 \\
EBXIID family ($g_1$) & 4 & 1093.00 & 2194.00 & 2214.19 \\
EBXIID family ($g_2$) & 2 & 1060.48 & $\mathbf{2124.96}$ & $\mathbf{2135.06}$ \\
EBXIID family ($g_3$) & 5 & 1091.23 & 2192.46 & 2217.70 \\
Log-normal & 2 & 1375.06 & 2754.12 & 2764.21 \\
Gamma & 2 & 1268.92 & 2541.84 & 2551.93 \\
BS & 2 & 1438.13 & 2880.26 & 2890.35 \\
S$t$BS & 3 & 1059.92 & 2125.84 & 2140.98 \\
SCNBS & 3 & 1069.10 & 2144.20 & 2159.34 \\
\hline
\end{tabular}
\caption{Roller data: log-likelihoods and AIC and BIC values}
\label{RollerTab}
\end{table}

An interesting feature of the roller data examined is that they show a negative skewness.
However, the above-mentioned distributions are all right-skewed.
In order to relate in a better way both data and the previous models, a transformation of data is analyzed.
It consists in to consider the observations $6-x$ where $x$ is an observed height belonging to the roller data.
Since the original data are lower than 6, the new data are always positive.

Considering the new data, all models presented in Table \ref{RollerTab} have been fitted to them.
Parameter estimates of most of models have been computed using the function \texttt{optimx}, excepting for the log-normal and gamma distributions.
For these two last cases the function \texttt{fitdist} has been used.
Table \ref{RollerGELS2} presents these estimates.
The parameter symbols given in \cite{BalakrishnanSauloLeao2017} have been adopted.

\begin{table}[!h]
\centering
\begin{tabular}{ccccccc}
\hline
$g$ & $a$ & $b$ & $c$ & $d$ & $\epsilon$ & $p$ \\
\hline
$g_0$ & 1 (fixed) & $7.021_{(0.173)}$ & $0.419_{(0.003)}$ & 0 (fixed) & \\
$g_1$ $^{(1)}$ & 1 (fixed) & $4.033_{(0.099)}$ & $0.195_{(0.001)}$ & 2 (fixed) & 2 (fixed) & $0.195_{(0.001)}$ \\
$g_2$ $^{(1)}$ & $0.683_{(0.065)}$ & $17.701_{(0.918)}$ & $0.484_{(0.005)}$ & 1 (fixed) & 1 (fixed) & $0.484_{(0.005)}$ \\
$g_3$ $^{(2)}$ & 1 (fixed) & $2.685_{(0.032)}$ & $0.087_{(0.002)}$ & $2.685_{(0.032)}$ & $2.685_{(0.032)}$ \\
\hline
\multicolumn{7}{l}{$^{(1)}$ Taken $c=p$} \\
\multicolumn{7}{l}{$^{(2)}$ Taken $b=d=\epsilon$}
\end{tabular}
\begin{tabular}{lrcccc}
\multicolumn{6}{c}{} \\
\hline
Model & \multicolumn{1}{c}{$\alpha$} & $\beta$ & $\lambda$ & $\nu$ & $\gamma$ \\
\hline
Log-normal & $0.869_{(0.007)}$ & $0.256_{(0.005)}$ \\
Gamma & $15.300_{(0.631)}$ & $6.206_{(0.260)}$ \\
BS & $0.259_{(0.005)}$ & $2.385_{(0.018)}$ \\
S$t$BS & $0.226_{(0.008)}$ & $2.383_{(0.017)}$ & $8.282_{(2.137)}$ & 0 (fixed) \\
SCNBS & $0.119_{(0.020)}$ & $2.363_{(0.045)}$ & $0.734_{(0.072)}$ & $0.020_{(0.046)}$ & $0.405_{(0.060)}$ \\
\hline
\end{tabular}
\caption{Fit of transformed roller data using models indicated in Table \ref{RollerTab}}
\label{RollerGELS2}
\end{table}

For the new fitted models, their AIC and BIC values, log-likelihoods and model parameter numbers are shown in 
Table \ref{RollerTab2}.
Comparing Tables \ref{RollerTab} and \ref{RollerTab2} some notorious differences appear.
First, all models now show better performances, which is in line with the previous observation that right-skewed distributions should fit right-skewed data, motivating data transformation.
Next, the favorite models have changed, but this selection varying according to the parsimony criteria.
Considering AIC values the recommended model would be the SCNBS distribution, whereas under BIC values would be the new model with $g=g_1$.
These two alternatives are contrary with respect to the heaviness of the distribution tails.
The first model being light-tailed and the second heavy-tailed.
Further, other heavy-tailed models show BIC values that are close to the one of the new model with $g=g_1$, so the new model with $g=g_0$ or $g=g_3$.

\begin{table}[!h]
\centering
\begin{tabular}{lcccc}
\hline
\multicolumn{1}{c}{Model} & $w$ & $-l$ & AIC & BIC \\
\hline
EBXIID family ($g_0$) & 2 & 1058.15 & 2120.31 & 2130.41 \\
EBXIID family ($g_1$) & 2 & 1058.08 & 2120.16 & $\mathbf{2130.26}$ \\
EBXIID family ($g_2$) & 3 & 1058.43 & 2122.86 & 2138.00 \\
EBXIID family ($g_3$) & 2 & 1058.14 & 2120.29 & 2130.39 \\
Log-normal & 2 & 1075.40 & 2154.81 & 2164.90 \\
Gamma & 2 & 1067.43 & 2138.87 & 2148.96 \\
BS & 2 & 1069.50 & 2143.01 & 2153.10 \\
S$t$BS & 3 & 1059.92 & 2124.44 & 2139.58 \\
SCNBS & 5 & 1053.33 & $\mathbf{2116.67}$ & 2141.90 \\
\hline
\end{tabular}
\caption{Transformed roller data: log-likelihoods and AIC and BIC values}
\label{RollerTab2}
\end{table}

\section{Discussion and Conclusion}
\label{sec6}

In this paper an extension of the Burr type XII (BXII) distribution has been proposed, taking advantage of the ability of this model for fitting large kinds of data.
This extension $F$ has great flexibility for fitting data thanks to the incorporation of a function $g$ that satisfies suitables conditions for guaranteeing $F$ is a probability distribution function.
The family of all these distributions has been called the extended Burr type XII distribution (EBXIID) family.
Analysis of some members of this family has shown that they can fit not only heavy-tailed data as the BXII distribution, but also light-tailed data.
Statistical properties of members of this family has been described.
Maximum likelihood estimation method has been proposed for estimating parameters of members of the EBXIID family.
Monte Carlo simulations keeping fixed some parameters have shown good performance of the parameter estimation method proposed.
Considering real data sets coming from diverse fields, some members of the EBXIID family have shown good performance with respect to competitors.

In this paper only a few members of the EBXIID family have been explored, but this family is very large and thus other members might be candidates for modeling specific data sets.
In a forthcoming paper, another model from this family will be explored in order to give a competitive representation of a particular real data set that have been analyzed with a generalization of the power law distribution with nonlinear exponent.
Also, in another forthcoming analysis, a well-known data set that have been modeled using different distributions, composite distributions among them, will be fitted using another member of this family.



\appendix

\section{Proofs}
\label{Proofs}

\begin{proof}[Proof of Proposition \ref{prop:20170220}:]
It is enough to note that
$$
\barF'(x)=-a\big(1+g(x)\big)^{-a-1}g'(x)\leq0,
$$
i.e. $\barF$ is decreasing; next,
$$
\lim_{x\to d^{+}}\barF(x)=\lim_{x\to d^{+}}\big(1+g(x)\big)^{-a}\leq1,
$$
and
$$
\lim_{x\to \infty}\barF(x)=\lim_{x\to \infty}\big(1+g(x)\big)^{-a}=0
$$
since $g(x)\to\infty$ as $x\to\infty$.
\end{proof}

\begin{proof}[Proof of Proposition \ref{PropMode}:]
Since $g$ defined in Table \ref{tab02} is twice differentiable, then $F$ defined in (\ref{defF}) is also twice differentiable.
Hence, the maximum $x_m$ of $F'$ ($\geq0$) is reached either in $(0,\infty)$ or at $x=0$, see e.g. \cite{RobertAAdams}.
In the first case, $x_m$ satisfies $F''(x_m)=0$, which implies, after straightforward computations, that
$$
(a+1)\big[g'(x_m)\big]^2=(1+g(x_m))g''(x_m).
$$
In the second case, we have $x_m=0=\inf\textrm{support}\,(g)$.
\end{proof}

\section{Relations For Estimating Parameters}
\label{RelationsForEstimatingParameters}

\subsection{$g_0$}

Conditions and information matrix for point and interval estimates for the parameters of the BII distribution can be found in e.g. \cite{WangKeatsZimmer1996}.

\subsection{$g_1$}

Let 
\begin{eqnarray*}
A(x) & := & \frac{\epsilon}{x}+\frac{p}{x+d}-\frac{1}{(x+d)\,\log(x+d)} \\
B(x) & := & 1-(a+1)\,\frac{g_1(x)}{1+g_1(x)} \\
C(x) & := & \frac{g_1(x)}{1+g_1(x)} \\
D(x) & := & p-\frac{1}{\log(x+d)} \\
E(x) & := & p-\frac{1+\log(x+d)}{\log^2(x+d)}.
\end{eqnarray*}

\begin{eqnarray*}
\frac{\partial l}{\partial a} & = & \sum_{k=1}^n\left(\frac{1}{a}-\log\big(1+g_1(x_k)\big)\right)\quad =\quad 0 \\
\frac{\partial l}{\partial b} & = & \frac{1}{b}\sum_{k=1}^n \left(1+B(x_k)\,\log g_1(x_k)\right)\quad =\quad 0 \\
\frac{\partial l}{\partial c} & = & \frac{b}{c}\sum_{k=1}^n B(x_k)\quad =\quad 0 \\
\frac{\partial l}{\partial d} & = & \sum_{k=1}^n\frac{1}{x_k+d}\,\left[b\,B(x_k)\,D(x_k)-\frac{E(x_k)}{(x_k+d)\,A(x_k)}\right]\quad =\quad 0 \\
\frac{\partial l}{\partial \epsilon} & = & \sum_{k=1}^n\left[b\,\log(x_k)\,B(x_k)+\frac{1}{x_k\,A(x_k)}\right]\quad =\quad 0 \\
\frac{\partial l}{\partial p} & = & \sum_{k=1}^n\left[b\,\log(x_k+d)\,B(x_k)+\frac{1}{(x_k+d)\,A(x_k)}\right]\quad =\quad 0
\end{eqnarray*}

\begin{eqnarray*}
\frac{\partial^2 l}{\partial a^2} & = & -\sum_{k=1}^n \frac{1}{a^2} \\
\frac{\partial^2 l}{\partial b\partial a} & = & -\frac{1}{b}\sum_{k=1}^n C(x_k)\,\log g_1(x_k) \\
\frac{\partial^2 l}{\partial c\partial a} & = & -\frac{b}{c}\sum_{k=1}^n C(x_k) \\
\frac{\partial^2 l}{\partial d\partial a} & = & -b\sum_{k=1}^n \frac{C(x_k)}{x_k+d}\,D(x_k) \\
\frac{\partial^2 l}{\partial \epsilon\partial a} & = & -b\sum_{k=1}^n C(x_k)\,\log x_k \\
\frac{\partial^2 l}{\partial p\partial a} & = & -b\sum_{k=1}^n C(x_k)\,\log (x_k+d) \\
\frac{\partial^2 l}{\partial b^2} & = & -\frac{1}{b^2}\sum_{k=1}^n \left[1+(a+1)\,\frac{C(x_k)\,\log^2 g_1(x_k)}{1+g_1(x_k)}\right] \\
\frac{\partial^2 l}{\partial c\partial b} & = & \frac{1}{c}\sum_{k=1}^n \left[1-(a+1)\,C(x_k)\,\left(1+\frac{\log g_1(x_k)}{1+g_1(x_k)}\right)\right] \\
\frac{\partial^2 l}{\partial d\partial b} & = & \sum_{k=1}^n \frac{D(x_k)}{x_k+d}\,\left[1-(a+1)\,C(x_k)\,\left(1+\frac{\log g_1(x_k)}{1+g_1(x_k)}\right)\right] \\
\frac{\partial^2 l}{\partial \epsilon\partial b} & = & \sum_{k=1}^n \log(x_k)\,\left[1-(a+1)\,C(x_k)\,\left(1+\frac{\log g_1(x_k)}{1+g_1(x_k)}\right)\right] \\
\frac{\partial^2 l}{\partial p\partial b} & = & \sum_{k=1}^n \log(x_k+d)\,\left[1-(a+1)\,C(x_k)\,\left(1+\frac{\log g_1(x_k)}{1+g_1(x_k)}\right)\right] \\
\frac{\partial^2 l}{\partial c^2} & = & -\frac{b}{c^2}\sum_{k=1}^n \left[1-(a+1)\,C(x_k)\,\left(1-\frac{b}{1+g_1(x_k)}\right)\right] \\
\frac{\partial^2 l}{\partial d\partial c} & = & -\frac{(a+1)\,b^2}{c}\sum_{k=1}^n \frac{C(x_k)\,D(x_k)}{(x_k+d)\,\big(1+g_1(x_k)\big)} \\
\frac{\partial^2 l}{\partial \epsilon\partial c} & = & -\frac{(a+1)\,b^2}{c}\sum_{k=1}^n \frac{C(x_k)\,\log x_k}{\big(1+g_1(x_k)\big)} \\
\frac{\partial^2 l}{\partial p\partial c} & = & -\frac{(a+1)\,b^2}{c}\sum_{k=1}^n \frac{C(x_k)\,\log (x_k+d)}{\big(1+g_1(x_k)\big)} \\
\frac{\partial^2 l}{\partial d^2} & = & -\sum_{k=1}^n \frac{1}{(x_k+d)^2}\left\{b\,B(x_k)\,E(x_k)+(a+1)\,b\,\frac{C(x_k)\,D^2(x_k)}{1+g_1(x_k)}+\frac{E^2(x_k)}{(x_k+d)^2\,A^2(x_k)}\right] \\
\frac{\partial^2 l}{\partial \epsilon\partial d} & = & -\sum_{k=1}^n \left[(a+1)\,b^2\frac{C(x_k)\,D(x_k)\,\log x_k}{(x_k+d)\,\big(1+g_1(x_k)\big)}-\frac{E(x_k)}{x_k\,(x_k+d)^2\,A^2(x_k)}\right] \\
\frac{\partial^2 l}{\partial p\partial d} & = & -\sum_{k=1}^n \left[-b\,\frac{B(x_k)}{x_k+d}+(a+1)\,b^2\frac{C(x_k)\,D(x_k)\,\log (x_k+d)}{(x_k+d)\,\big(1+g_1(x_k)\big)}+\frac{E(x_k)}{(x_k+d)^3\,A^2(x_k)}+\frac{1}{(x_k+d)^2\,A(x_k)}\right] \\
\frac{\partial^2 l}{\partial \epsilon^2} & = & -\sum_{k=1}^n \left[(a+1)\,b^2\,\frac{C(x_k)\,\log^2 x_k}{1+g_1(x_k)}+\frac{1}{x_k^2A^2(x_k)}\right] \\
\frac{\partial^2 l}{\partial p\partial\epsilon} & = & -\sum_{k=1}^n\left[(a+1)\,b^2\,\frac{C(x_k)\,\log x_k\,\log (x_k+d)}{1+g_1(x_k)}+\frac{1}{x_k\,(x_k+d)\,A^2(x_k)}\right] \\
\frac{\partial^2 l}{\partial p^2} & = & -\sum_{k=1}^n\left[(a+1)b^2\frac{C(x_k)\,\log^2 (x_k+d)}{1+g_1(x_k)}+\frac{1}{(x_k+d)^2\,A^2(x_k)}\right]
\end{eqnarray*}

\subsection{$g_2$}

Let
\begin{eqnarray*}
A(x) & := & b\left(\frac{\epsilon}{x}-\frac{1}{(x+d)\,\log(x+d)}\right)+p\,x^{p-1} \\
B(x) & := & 1-(a+1)\,\frac{g_2(x)}{1+g_2(x)} \\
C(x) & := & \frac{g_2(x)}{1+g_2(x)} \\
D(x) & := & \log\left(\frac{c\,x^\epsilon}{\log(x+d)}\right) \\
E(x) & := & \frac{\epsilon}{x}-\frac{1}{(x+d)\,\log(x+d)} \\
F(x) & := & (x+d)\,\log(x+d) \\
G(x) & := & \frac{1+\log(x_k+d)}{(x+d)^2\,\log^2(x+d)}.
\end{eqnarray*}

\begin{eqnarray*}
\frac{\partial l}{\partial a} & = & \sum_{k=1}^n\left(\frac{1}{a}-\log\big(1+g_2(x_k)\big)\right)\quad =\quad 0 \\
\frac{\partial l}{\partial b} & = & \sum_{k=1}^n \left[D(x_k)\,B(x_k)+\frac{E(x_k)}{A(x_k)}\right]\quad =\quad 0 \\
\frac{\partial l}{\partial c} & = & \frac{b}{c}\sum_{k=1}^n B(x_k)\quad =\quad 0 \\
\frac{\partial l}{\partial d} & = & -b\sum_{k=1}^n \frac{1}{F(x_k)}\,\left[B(x_k)-\frac{G(x_k)}{A(x_k)}\right]\quad =\quad 0 \\
\frac{\partial l}{\partial \epsilon} & = & b\sum_{k=1}^n\left[B(x_k)\,\log(x_k)+\frac{1}{A(x_k)\,x_k}\right]\quad =\quad 0 \\
\frac{\partial l}{\partial p} & = & \sum_{k=1}^n x_k^{p}\,\left[B(x_k)\,\log x_k+\frac{1+p\,\log x_k}{x_k\,A(x_k)}\right]\quad =\quad 0
\end{eqnarray*}

\begin{eqnarray*}
\frac{\partial^2 l}{\partial a^2} & = & -\sum_{k=1}^n \frac{1}{a^2} \\
\frac{\partial^2 l}{\partial b\partial a} & = & -\sum_{k=1}^n C(x_k)\,D(x_k) \\
\frac{\partial^2 l}{\partial c\partial a} & = & -\frac{b}{c}\sum_{k=1}^n C(x_k) \\
\frac{\partial^2 l}{\partial d\partial a} & = & -b\sum_{k=1}^n \frac{C(x_k)}{F(x_k)} \\
\frac{\partial^2 l}{\partial \epsilon\partial a} & = & -b\sum_{k=1}^n C(x_k)\,\log x_k \\
\frac{\partial^2 l}{\partial p\partial a} & = & -\sum_{k=1}^n x_k^{p}\,C(x_k)\,\log x_k \\
\frac{\partial^2 l}{\partial b^2} & = & -\sum_{k=1}^n \left[(1+a)\,\frac{C(x_k)\,D^2(x_k)}{1+g_2(x_k)}+\frac{E^2(x_k)}{A^2(x_k)}\right] \\
\frac{\partial^2 l}{\partial c\partial b} & = & \frac{1}{c}\sum_{k=1}^n \left[B(x_k)-(1+a)\,b\,\frac{C(x_k)\,D(x_k)}{1+g_2(x_k)}\right] \\
\frac{\partial^2 l}{\partial d\partial b} & = & -\sum_{k=1}^n \frac{1}{F(x_k)}\,\left[B(x_k)-(1+a)\,b\,\frac{C(x_k)\,D(x_k)}{1+g_2(x_k)}+\frac{1+\log(x_k+d)}{A(x_k)\,F(x_k)}\,\left(b\,\frac{E(x_k)}{A(x_k)}-1\right)\right] \\
\frac{\partial^2 l}{\partial \epsilon\partial b} & = & \sum_{k=1}^n \left[\log(x_k)\,\left(B(x_k)-(1+a)\,\frac{\log(x_k)\,C(x_k)\,D(x_k)}{1+g_2(x_k)}\right)+\frac{1}{x_k\,A(x_k)}-\frac{E(x_k)}{x_k\,A^2(x_k)}\right] \\
\frac{\partial^2 l}{\partial p\partial b} & = & -\sum_{k=1}^n x_k^{p}\,\left[(1+a)\,\frac{C(x_k)\,D(x_k)\,\log x_k}{1+g_2(x_k)}+\frac{\big(1+p\,\log x_k\big)\,E(x_k)}{x_k\,A^2(x_k)}\right] \\
\frac{\partial^2 l}{\partial c^2} & = & -\frac{b}{c^2}\sum_{k=1}^n \left[B(x_k)+(1+a)\,b\,\frac{C(x_k)}{1+g_2(x_k)}\right] \\
\frac{\partial^2 l}{\partial d\partial c} & = & \frac{(1+a)\,b^2}{c}\sum_{k=1}^n \frac{C(x_k)}{(x_k+d)\,\big(1+g_2(x_k)\big)\,\log(x_k+d)} \\
\frac{\partial^2 l}{\partial \epsilon\partial c} & = & -\frac{(1+a)\,b^2}{c}\sum_{k=1}^n \frac{C(x_k)\,\log x_k}{1+g_2(x_k)} \\
\frac{\partial^2 l}{\partial p\partial c} & = & -\frac{(1+a)\,b}{c}\sum_{k=1}^n \frac{x_k^{p}\,C(x_k)\,\log x_k}{1+g_2(x_k)} \\
\frac{\partial^2 l}{\partial d^2} & = & b\sum_{k=1}^n \left[B(x_k)\,G(x_k)-(1+a)\,b\,\frac{C(x_k)}{F^2(x_k)\,\big(1+g_2(x_k)\big)}-b\,\frac{G^2(x_k)}{A^2(x_k)}+\frac{1}{A(x_k)\,(x_k+d)\,F^2(x_k)}-\frac{G^2(x_k)}{A(x_k)\,F(x_k)}\right] \\
\frac{\partial^2 l}{\partial \epsilon\partial d} & = & b^2\,\sum_{k=1}^n \left[(1+a)\,\frac{C(x_k)\,\log x_k}{\big(1+g_2(x_k)\big)\,F(x_k)}-\frac{G(x_k)}{x_k\,A^2(x_k)}\right] \\
\frac{\partial^2 l}{\partial p\partial d} & = & b\,\sum_{k=1}^n x_k^{p}\,\left[(1+a)\,\frac{C(x_k)\,\log x_k}{\big(1+g_2(x_k)\big)\,F(x_k)}-\frac{G(x_k)\,\big(1+\log(x_k+d)\big)}{x_k\,A^2(x_k)}\right] \\
\frac{\partial^2 l}{\partial \epsilon^2} & = & -b^2\sum_{k=1}^n \left[(1+a)\,\frac{C(x_k)\,\log^2 x_k}{1+g_2(x_k)}+\frac{1}{x_k^2A^2(x_k)}\right] \\
\frac{\partial^2 l}{\partial p\partial\epsilon} & = & -b\sum_{k=1}^n x_k^{p}\,\left[(a+1)\,\frac{C(x_k)\,\log^2 x_k}{1+g_2(x_k)}+\frac{1+p\,\log x_k}{x_k^2\,A^2(x_k)}\right] \\
\frac{\partial^2 l}{\partial p^2} & = & \sum_{k=1}^n x^{p}\left[\log^2 x_k\,\left(B(x_k)-(a+1)\,\frac{x_k^p\,C(x_k)}{1+g_2(x_k)}\right)+\frac{x_k^{p}\,\big(1+p\,\log x_k\big)^2}{x_k^2\,A^2(x_k)}+\frac{\big(2+\log x_k\big)\,\log x_k}{x_k\,A(x_k)}\right]
\end{eqnarray*}

\subsection{$g_3$}

Let
\begin{eqnarray*}
A(x) & := & \frac{\epsilon}{x}+B(x) \\
B(x) & := & \frac{1}{(x+d+1)\,\log(x+d+1)}-\frac{1}{(x+d)\,\log(x+d)} \\
C(x) & := & -\frac{\log(x+d+1)+1}{(x+d+1)^2\,\log^2(x+d+1)}-\frac{\log(x+d)+1}{(x+d)^2\,\log^2(x+d)} \\
D(x) & := & 1-(a+1)\,\frac{g_3(x)}{1+g_3(x)} \\
E(x) & := & \frac{g_3(x)}{1+g_3(x)} \\
F(x) & := & 2\,\frac{\big(1+\log x\big)^2}{x^3\,\log^3 x}-\frac{1}{x^3\,\log^2 x}.
\end{eqnarray*}

\begin{eqnarray*}
\frac{\partial l}{\partial a} & = & \sum_{k=1}^n\left(\frac{1}{a}-\log\big(1+g_3(x_k)\big)\right)\quad =\quad 0 \\
\frac{\partial l}{\partial b} & = & \frac{1}{b}\,\sum_{k=1}^n \left[1+D(x_k)\,\log g_3(x_k)\right]\quad =\quad 0 \\
\frac{\partial l}{\partial c} & = & \frac{b}{c}\sum_{k=1}^n D(x_k)\quad =\quad 0 \\
\frac{\partial l}{\partial d} & = & \sum_{k=1}^n\left[b\,B(x_k)\,D(x_k)-\frac{C(x_k)}{A(x_k)}\right]\quad =\quad 0 \\
\frac{\partial l}{\partial \epsilon} & = & \sum_{k=1}^n\left[b\,D(x_k)\,\log(x_k)+\frac{1}{x_k\,A(x_k)}\right]\quad =\quad 0
\end{eqnarray*}

\begin{eqnarray*}
\frac{\partial^2 l}{\partial a^2} & = & -\sum_{k=1}^n \frac{1}{a^2} \\
\frac{\partial^2 l}{\partial b\partial a} & = & -\frac{1}{b}\sum_{k=1}^n E(x_k)\,\log g_3(x_k) \\
\frac{\partial^2 l}{\partial c\partial a} & = & -\frac{b}{c}\sum_{k=1}^n E(x_k) \\
\frac{\partial^2 l}{\partial d\partial a} & = & -b\sum_{k=1}^n E(x_k)\,B(x_k) \\
\frac{\partial^2 l}{\partial \epsilon\partial a} & = & -b\sum_{k=1}^n E(x_k)\,\log x_k \\
\frac{\partial^2 l}{\partial b^2} & = & -\frac{1}{b^2}\sum_{k=1}^n \left[1+(a+1)\,\frac{E(x_k)\,\log^2g_3(x_k)}{1+g_3(x_k)}\right] \\
\frac{\partial^2 l}{\partial c\partial b} & = & \frac{1}{c}\sum_{k=1}^n \left[D(x_k)-(a+1)\,\frac{E(x_k)\,\log g_3(x_k)}{1+g_3(x_k)}\right] \\
\frac{\partial^2 l}{\partial d\partial b} & = & \sum_{k=1}^n B(x_k)\,\left[D(x_k)-(a+1)\,\frac{E(x_k)\,\log g_3(x_k)}{1+g_3(x_k)}\right] \\
\frac{\partial^2 l}{\partial \epsilon\partial b} & = & \sum_{k=1}^n \log x_k\,\left[D(x_k)-(a+1)\,\frac{E(x_k)\,\log g_3(x_k)}{1+g_3(x_k)}\right] \\
\frac{\partial^2 l}{\partial c^2} & = & -\frac{b}{c^2}\sum_{k=1}^n \left[D(x_k)+(a+1)\,b\,\frac{E(x_k)}{1+g_3(x_k)}\right] \\
\frac{\partial^2 l}{\partial d\partial c} & = & -\frac{(a+1)\,b^2}{c}\sum_{k=1}^n \frac{E(x_k)\,B(x_k)}{1+g_3(x_k)} \\
\frac{\partial^2 l}{\partial \epsilon\partial c} & = & -\frac{(a+1)\,b^2}{c}\sum_{k=1}^n \frac{E(x_k)\,\log x_k}{1+g_3(x_k)} \\
\frac{\partial^2 l}{\partial d^2} & = & \sum_{k=1}^n \left[b\,C(x_k)\,D(x_k)-(a+1)\,b^2\frac{E(x_k)\,B^2(x_k)}{1+g_3(x_k)}-\frac{C^2(x_k)}{A^2(x_k)}+\frac{F(x_k+d+1)-F(x_k+d)}{A(x_k)}\right] \\
\frac{\partial^2 l}{\partial \epsilon\partial d} & = & -\sum_{k=1}^n \left[(a+1)\,b^2\frac{E(x_k)\,B(x_k)\,\log g_3(x_k)}{1+g_3(x_k)}+\frac{C(x_k)}{x_k\,A^2(x_k)}\right] \\
\frac{\partial^2 l}{\partial \epsilon^2} & = & -\sum_{k=1}^n \left[(a+1)\,b^2\,\frac{E(x_k)\,\log^2 x_k}{1+g_3(x_k)}+\frac{1}{x^2_k\,A^2(x_k)}\right]
\end{eqnarray*}

\end{document}